\documentclass[a4paper,11pt]{article}
\usepackage[utf8]{inputenc}

\usepackage{color}

\usepackage{amssymb}
\usepackage{amsmath}
\usepackage{amsfonts}
\usepackage{amsthm}
\usepackage{url}

\usepackage[dvips]{epsfig}
\usepackage[dvips]{graphics}

\usepackage{algorithm}
\usepackage{algorithmic}

\usepackage[all,knot,poly,arc]{xy}

\newcommand{\cA}{\mbox{$\cal A$}}

\newcommand{\NN}{{\Bbb N}}
\newcommand{\RR}{{\Bbb R}}

\newcommand{\ZZ}{{\Bbb Z}}

\newcommand{{\uk}}{\mbox{$\underline{k}$}}

\def\nod(#1,#2){\put(#1,#2){\circle*{.125}}
\put(#1,#2){\makebox(0,0.5){{\small$#2$}}}}%
\def\rod(#1,#2){\put(#1,#2){\circle*{.2}}}
\def\NOD(#1,#2)#3{\put(#1,#2){\circle*{.2}}\put(#1,#2){\makebox(0,0.8){{\small$#3$}}}}

\def\EXX{{\hfill{$\diamondsuit$}}}

\newcounter{exampleNo}

\newtheorem{theorem}{Theorem}[section]

\newtheorem{proposition}[theorem]{Proposition}
\newtheorem{corollary}[theorem]{Corollary}

\newenvironment{example}[1][Example \arabic{exampleNo}.]{\begin{trivlist}\refstepcounter{exampleNo}
\item[\hskip \labelsep {\bfseries #1}]}{\end{trivlist}}

\title{On Stochastic Automata over Monoids}
\author{Merve Nur Cakir and Karl-Heinz Zimmermann\footnote{Email: k.zimmermann@tuhh.de}\\
Department of Computer Engineering \\
Hamburg University of Technology\\
21071 Hamburg, Germany}

\begin{document}
\maketitle
\begin{abstract}
In this paper, stochastic automata over monoids as input sets are studied.
The well-definedness of these automata requires an extension postulate that
replaces the inherent universal property of free monoids.
As a generalization of Turakainen's result, it will be shown 
that the generalized automata over monoids have the same acceptance power as their stochastic counterparts.
The key to homomorphisms is a commuting property between the monoid homomorphism of input states and 
the monoid homomorphism of transition matrices.
Closure properties of the languages accepted by stochastic automata over monoids are studied.
%They are in a sense parallel to the closure properties of the languages accepted by monoidal automata. 
\end{abstract}
\medskip

\mbox{\bf AMS Subject Classification:} 68Q70, 68Q87, 20M35
\medskip

\mbox{\bf Keywords:} Stochastic automaton, formal language, monoid, homomorphism, closure properties.

\section{Introduction}

The theory of discrete stochastic systems has been initiated by the work of Shannon~\cite{shannon} and von Neumann~\cite{neumann}.
While Shannon has considered memory-less communication channels and their generalization by introducing states,
von Neumann has studied the synthesis of reliable systems from unreliable components.
The fundamental work of Rabin and Scott~\cite{rscott} about deterministic finite-state automata 
has led to two generalizations.
First, the generalization of transition functions to conditional distributions 
studied by Carlyle~\cite{carl} and Starke~\cite{starke}.
This in turn yields a generalization of discrete-time Markov chains in which the chains are governed 
by more than one transition probability matrix.
Second, the generalization of regular sets by introducing stochastic automata as described by Rabin~\cite{rabin}. 
%Stochastic automata are well-investigated~\cite{claus}.

By the work of Turakainen~\cite{tura69}, stochastic acceptors can be 
viewed equivalently as generalized automata in which the ''probability'' is neglected.
This leads to a more accessible approach to stochastic automata~\cite{claus}.

On the other hand, the class of nondeterministic automata~\cite{salomaa} can be generalized to monoidal automata,
where the input alphabet corresponds to an arbitrary monoid instead of a free monoid~\cite{kufi,mihov}.
This leads to the class of monoidal automata whose languages are closed under a smaller set of operations 
when compared with regular languages.

In this paper, a unification of generalized automata and monoidal automata, called monoidal generalized automata, is studied. 
In view of the well-definedness of these automata, an extension postulate is necessary 
to replace the inherent universal property of free monoids.
As a generalization of Turakainen's result, it will be shown 
that the monoidal generalized automata have the same acceptance power as their stochastic counterparts.
Moreover, the key to homomorphisms is a commuting property 
between the monoid homomorphism of input states and the monoid homomorphism of transition matrices.
Closure properties of the languages accepted by monoidal generalized automata are studied.
They are in a sense parallel to the closure properties of the languages accepted by monoidal automata. 
The text is largely self-contained and can be read with moderate preknowledge in stochastics and formal languages.

%-------------------------------------------------
\section{Stochastic Automata}
Stochastic finite-state automata are a generalization of the non-deterministic finite-state automata~\cite{claus}.
%The languages recognized by stochastic automata are called stochastic languages.
%We will see that the class of stochastic languages is uncountable and includes the regular languages.
%For this, we assume familiarity with the basic concepts of regular languages 
%as well as deterministic and nondeterministic finite automata.

A {\em stochastic automaton\/} is a quintuple 
$$\cA = (S,\Sigma,\{P(x)\mid x\in\Sigma\}, \pi,f),$$ 
where
\begin{itemize}
\item $S$ is the non-empty finite set of {\em states},
\item $\Sigma$ is the alphabet of {\em input symbols},
\item $P$ is a collection of (row-) stochastic $n\times n$ matrices $P(x)$, $x\in\Sigma$,
where $n$ is the number of states,
\item $\pi$ is the {\em initial distribution}\index{initial distribution} 
of the state set written as row vector,
\item $f$ is a binary column vector of length $n$ called {\em final state vector}\index{final state vector}.
\end{itemize}
Note that if the state set is $S=\{s_1,\ldots,s_n\}$ 
and the final state vector is $f=(f_1,\ldots,f_n)^T$, then 
$F = \{s_i\mid f_i=1\}$ is the {\em final state set}\index{final state set}.

Note that $\Sigma^*$ is the free monoid over the alphabet $\Sigma$ and so by the 
universal property of free monoids,
there exists a unique monoid homomorphism $P:\Sigma^*\rightarrow\RR^{n\times n}$ 
from the free monoid $(\Sigma^*,\circ,\epsilon)$ 
to the multiplicative monoid of $n\times n$ real-valued matrices $(\RR^{n\times n},\cdot,I_n)$
that extends the mapping $P:\Sigma\rightarrow\RR^{n\times n}$ given by the 
automaton~\cite{cliff,mihov}.
Thus for each word $u=x_1\ldots x_k\in\Sigma^*$,
the associated matrix is $P(u) = P(x_1)\cdots P(x_k)$.
In particular, if $u=\epsilon$ is the empty word, then $P(\epsilon)=I_n$ is the $n\times n$ unit matrix.
Furthermore, the $(i,j)$th element $p(s_j\mid u, s_i)$ of the matrix $P(u)$
is the transition probability that the automaton enters state $s_j$ when starting in state $s_i$ and reading the word $u$. 

%A stochastic automaton $A=(S,\Sigma,P,\pi,f)$ is {\em deterministic\/}\index{stochastic acceptor!deterministic} 
%if each matrix $P(a)$, $a\in\Sigma$, has in each row exactly one entry~1.
%Note that a deterministic stochastic automaton is a nondeterministic automaton in the usual sense if the initial state distribution is a single state.
%%\begin{example}
%%Consider the deterministic automaton $A = (\{s_1,s_2\},\{a,b\}, P, \pi,f)$ with
%%$$ P(a) = \left( \begin{array}{cc} 1&0\\ 0 & 1 \end{array} \right),
%%\quad P(b) = \left( \begin{array}{cc} 0&1\\ 1 & 0 \end{array} \right),
%%\quad \pi = (1,0), 
%%\quad\mbox{and}\quad f= \left( \begin{array}{l} 0\\ 1 \end{array} \right).  $$
%%Then for input words of length~3,
%%\begin{eqnarray*}
%%\begin{array}{lll}
%%\pi P(aaa) f = 0, && \pi P(baa) f = 1,\\
%%\pi P(aab) f = 1, && \pi P(bab) f = 0,\\
%%\pi P(aba) f = 1, && \pi P(bba) f = 0,\\
%%\pi P(abb) f = 0, && \pi P(bbb) f = 1.\\
%%\end{array}
%%\end{eqnarray*}
%%\EXX
%%\end{example}

Let $\cA=(S,\Sigma,\{P(x)\mid x\in\Sigma\},\pi,f)$ 
be a stochastic automaton and let $\lambda$ be a real number with $0\leq \lambda\leq 1$.
The set 
\begin{eqnarray}
L(\cA,\lambda) = \{u\in \Sigma^*\mid \pi P(u)f > \lambda\}
\end{eqnarray}
is the {\em language\/}\index{language} of the automaton $\cA$ 
w.r.t.\ the {\em cut point\/}\index{cut point} $\lambda$.
%A language $L\subseteq \Sigma^*$ is {\em $\lambda$-stochastic} 
%if there is a stochastic finite-state automaton $\cA$ 
%such that $L=L(\cA,\lambda)$.
A subset $L\subseteq \Sigma^*$ is a {\em stochastic automaton language} 
if there exists a stochastic automaton $\cA$ and a cut point $\lambda$ with $0\leq \lambda\leq 1$
such that $L=L(\cA,\lambda)$.
The $m$-adic languages provide a class of stochastic automaton languages 
which contains properly the class of regular languages~\cite{claus}.
\begin{example}
Let $m\geq 2$ be an integer.
Put $\Sigma=\{0,\ldots,m-1\}$.
The stochastic automaton $\cA= (\{s_1,s_2\}, \Sigma,  \{P(x)\mid x\in\Sigma\} , \pi, f)$ 
given by
$$P(x) = 
\left(\begin{array}{cc} 1-\frac{x}{m} &\frac{x}{m} \\ 1-\frac{x+1}{m} &\frac{x+1}{m} \end{array}\right), \quad x\in\Sigma, 
$$
$\pi =(1,0),$ and $f = {0\choose 1}$
is called {\em $m$-adic acceptor}.

A word $u=x_1\ldots x_k\in\Sigma^*$ lies in $L(\cA,\lambda)$ iff $\pi P(u) f>\lambda$, 
i.e., the $(1,2)$-entry of the matrix $P(u)$ is larger than $\lambda$.
This element is the $m$-adic representation $0.x_k\ldots x_1$ of $u$.
Thus the language accepted by the automaton $\cA$ w.r.t.\ the cut point $\lambda$ is 
$$L(\cA,\lambda) = \{x_1\ldots x_k\in\Sigma^*\mid 0.x_k\ldots x_1 > \lambda\}.$$
In particular, the language $L(\cA,\lambda)$ is regular iff $\lambda$ is a rational number.

Furthermore, if $\lambda\ne \lambda'$, then $L(\cA,\lambda)\ne L(\cA,\lambda')$ 
and so the class of stochastic automaton languages accepted by $\cA$ for different cut points 
is nondenumerable.
\EXX
\end{example}
%
%By Prop.~\ref{p-SA-rat}, each regular language is stochastic but not vice versa.
%If the cut point $\lambda$ is not rational, one obtains a stochastic language which is not regular. 
%For this, note that the set of algorithms written for a Turing machine is countable, while the set of real numbers is uncountable.
%%Moreover, for distinct cut points $\lambda_1$ and $\lambda_2$, the corresponding $p$-adic stochastic acceptors yield different stochastic languages.  
%Thus the number of stochastic languages is uncountable.
%But the number of programs written for a Turing machine is countable and hence there are stochastic languages that cannot be accepted by a Turing machine.
%This example shows that there are stochastic automata $A$, 
%whose transition matrices and initial distributions are given by rational numbers, 
%and rational cut points $\lambda$, but the language $L(A,\lambda)$ is not regular.

%%%%%%%%%%%%%%%%%%%%%%%%%%%%%%%%%%%%%%%%%%%%%%%%%%%%%%%%%%%%%%%%%%%%%%%%%%%%%%%%%%%
\section{Generalized Automata} 
The definition of stochastic finite-state automata can be generalized 
by dropping the restrictions imposed by probability~\cite{claus, tura69}.

A {\em generalized automaton\/} is a quintuple 
$$\cA = (S,\Sigma, \{Q(x)\mid x\in \Sigma\}, \pi,f),$$ 
where
\begin{itemize}
\item $S$ is the non-empty finite set of {\em states},
\item $\Sigma$ is the alphabet of {\em input symbols},
\item $Q$ is a collection of $n\times n$ matrices $Q(x)$, $x\in \Sigma$, 
where $n$ is the number of states,
\item $\pi\in\RR^n$ is the {\em initial vector} written as row vector, and
\item $f\in\RR^n$ is the {\em final vector} written as column vector.
\end{itemize}
%Note that if the state set is $S=\{s_1,\ldots,s_n\}$ and the final state vector is $f=(f_1,\ldots,f_n)^t$, then 
%$F = \{s_i\mid f_i=1\}$ is the {\em final state set}\index{final state set}.
As already noticed, $\Sigma^*$ is the free monoid over the alphabet $\Sigma$ and so by the 
universal property of free monoids,
there exists a unique monoid homomorphism $Q:\Sigma^*\rightarrow\RR^{n\times n}$ from the free monoid $(\Sigma^*,\circ,\epsilon)$ 
to the multiplicative monoid of $n\times n$ real-valued matrices $(\RR^{n\times n},\cdot,I_n)$
that extends the mapping $Q:\Sigma\rightarrow\RR^{n\times n}$ given by the automaton.
Thus for each word $u=x_1\ldots x_k\in\Sigma^*$, the associated matrix is $Q(u) = Q(x_1)\cdots Q(x_k)$.
In particular, if $u=\epsilon$ is the empty word, then $Q(\epsilon)=I_n$ is the $n\times n$ unit matrix.
%Note that $(i,j)$th element $p(s_j\mid u, s_i)$ of the matrix $P(u)$
%is the transition probability that the automaton enters state $s_j$ when starting in state $s_i$ and reading the word $u$. 
%Each input word $u=x_1\ldots x_k\in\Sigma^*$ is associated with the matrix $Q(u) = Q(x_1)\cdots Q(x_k)$.
%In particular, if $u=\epsilon$ is the empty word, $Q(\epsilon)=I_n$ is the $n\times n$ unit matrix.

Let $\cA = (S,\Sigma, \{Q(x)\mid x\in \Sigma\}, \pi,f)$ be a 
generalized automaton and $\lambda$ be a real number.
The set
\begin{eqnarray}
L(\cA,\lambda) =\{u\in\Sigma^*\mid \pi Q(u) f>\lambda\}
\end{eqnarray}
is the {\em language} accepted by the automaton $\cA$ w.r.t.\ the {\em cut point\/} $\lambda$.
A subset $L$ of $\Sigma^*$ is called a {\em generalized automaton language} 
if there exists a generalized automaton $\cA$
and a real number $\lambda$ such that $L=L(\cA,\lambda)$.

This generalization of stochastic automaton languages 
does not lead to a larger class of languages~\cite{tura69}.
\begin{theorem}[Turakainen]\label{p-tura0}
Each generalized automaton language is a stochastic automaton language.
\end{theorem}

The following proposition provides a characterization of stochastic automaton languages 
which does not use the notion of automaton~\cite{claus}.
This description has first been used in two seminal papers~\cite{fliess, schberger}.

\begin{theorem}[Matrix Characterization]\label{p-claus}
A subset $L$ of\/ $\Sigma^*$ is a stochastic automaton language 
iff there exists a collection $\{Q(x)\mid x\in\Sigma\}$
of $n\times n$ matrices for some $n\geq 1$ such that
for each non-empty word $u=x_1\ldots x_k \in\Sigma^*$, %$u\ne\epsilon$,
$$u\in L\quad\Longleftrightarrow\quad (Q(u))_{1,n}>0,$$
where $(Q(u))_{1,n}$ is the $(1,n)$-entry of the matrix $Q(u)$.
\end{theorem}

\begin{example}\label{e-cos}
Let $\varphi$ be a real number. 
Consider the rotation matrix
$$R_\varphi = \left( \begin{array}{rc} \cos(2\pi\varphi) & \sin(2\pi\varphi) \\ -\sin(2\pi\varphi) & \cos(2\pi\varphi) \end{array}\right).  $$
Then for each integer $n\geq 1$,
$$R_\varphi^n=\left(\begin{array}{rc} \cos(2\pi n\varphi) & \sin(2\pi n\varphi) \\ -\sin(2\pi n\varphi) & \cos(2\pi n\varphi)\end{array}\right).$$
Take the alphabet $\Sigma=\{x\}$.
Then by Prop.~\ref{p-claus},
the stochastic automaton language $L=L_\varphi \subseteq \Sigma^*$ contains the non-empty words $x^n$ 
whenever $\sin(2\pi n\varphi)>0$.
For instance, if $\varphi=30^o$, then $L$ will contain the non-empty words $x$ and $x^{4i},x^{4i+1}$ for each $i\geq 1$.
Note that if the initial vector $\pi$ and the final vector $f$ are specified, 
the value of $\pi f$ will determine whether the empty word lies in $L$.
\EXX
\end{example}

%----------------------------------------------------------------
Stochastic automaton languages are closed under several set-theoretic operations~\cite{claus}.
\begin{proposition}[Closure Properties]\label{p-closed}
Let $\Sigma,\Omega$ be alphabets, let $L,L_1,L_2$ be stochastic automaton languages over $\Sigma$,
and let $R$ be a regular language over $\Sigma$.
\begin{itemize}
\item
The mirror image of\/ $L$ is a stochastic automaton language.
\item 
$L\cap R$, $L\cup R$, and $L\setminus R$ are stochastic automaton languages.
\item
If\/ $\Sigma=\{x\}$ is a singleton set, the complement $\overline L = \{x\}^*\setminus L$ is a stochastic
automaton language.
\item
$L_1\cap L_2$, $L_1\cup L_2$, $L_1\circ L_2$ (product or concatenation) 
and $L^*$ (Kleene star) are generally not stochastic automaton languages.
\item 
If\/ $\phi:\Sigma^*\rightarrow\Omega^*$ is a monoid homomorphism, 
the image $\phi(L)$ is generally not a stochastic automaton language.
However, if\/ $\Sigma=\{x\}$ is a singleton set, the image $\phi(L)$ is a stochastic automaton language.
\end{itemize}
\end{proposition}

\begin{example} 
The class of stochastic automaton languages is not closed under union and intersection. % and complement.
To see this, consider the generalized automaton languages $L_\varphi$ from Ex.~\ref{e-cos}. 
Take real numbers $\varphi_1$ and $\varphi_2$ which are linearly independent over the rationals, i.e.,
there exist no rational numbers $r,r_1,r_2$ such that $r+r_1\varphi_1+r_2\varphi_2=0$.
Then the set $L_{\varphi_1}\cup L_{\varphi_2} \subseteq \{x\}^*$ 
is not a stochastic automaton language~\cite{fliess}.

By De Morgan's law, 
$L_{\varphi_1}\cup L_{\varphi_2} = \overline{\overline L_{\varphi_1}\cap \overline L_{\varphi_2} }$
and by Prop.~\ref{p-closed}, the complement $\overline L$ of a stochastic language $L\subseteq \{x\}^*$ 
is stochastic.
Hence, the intersection 
$\overline L_{\varphi_1}\cap \overline L_{\varphi_2}$
cannot be a stochastic automaton language.
\EXX 
\end{example}

Stochastic automaton languages are not closed under set-theoretic complement.
The opposite holds for isolated cut points~\cite{claus}.
A cut point $\lambda$ is {\em isolated\/} for a stochastic or generalized automaton 
$\cA$ if there exists a real number $\delta>0$ such that for all words $u\in\Sigma^*$, 
\begin{eqnarray}
|\lambda - \pi Q(u) f| \geq \delta.
\end{eqnarray}
\begin{proposition}[Closure under Complement]\label{p-com}
Let $L$ be a stochastic automaton language.
If\/ $L$ is accepted by the stochastic automaton $\cA$ w.r.t.\ 
the cut point $\lambda$ and $\lambda$ is isolated for~$\cA$, 
the complement $\overline L = \Sigma^*\setminus L$ is also a stochastic automaton language.
\end{proposition}

%%%%%%%%%%%%%%%%%%%%%%%%%%%%%%%%%%%%%%%%%%%%%%%%%%%%%%%%%%%%%%%%%%%%%%%%%%%%%%%%%%%
\section{Monoidal Automata} 

Monoidal finite-state automata are introduced as a generalization of classical finite-state automata~\cite{kufi,mihov}.
They are defined over an arbitrary monoid as input set instead of the free monoid over an alphabet.

A {\em monoidal automaton} %or {\em $M$-automaton} 
is a quintuple 
$$\cA=(S,M,I,F,\Delta),$$ 
where
\begin{itemize}
\item $S$ is the non-empty finite set of {\em states},
\item $(M,\circ,e)$ is a finitely-generated monoid, where $M$ is the set of {\em input symbols},
\item $I\subseteq S$ is the set of {\em initial states}, 
\item $F\subseteq S$ is the set of {\em final states}, and 
\item $\Delta \subseteq S\times M\times S$ is a finite set called the {\em transition relation}. 
%where $G_M$ is a generating set of $M$.
\end{itemize}
Triples $(s,x,s')\in\Delta$ are called {\em transitions}.
The transition $(s,x,s')\in\Delta$ {\em begins\/} in state $s$, {\em ends\/} in state $s'$ 
and has {\em label\/} $x$.

Let $\cA=(S,M,I,F,\Delta)$ be a monoidal automaton.
A {\em proper path\/} in $\cA$ is a finite sequence of $k\geq 1$ transitions
\begin{eqnarray}\label{e-path}
\pi = (s_0,x_1,s_1), (s_1,x_2,s_2), \ldots ,(s_{k-1},x_k,s_k),
\end{eqnarray}
where $(s_{j-1},x_j,s_j)\in\Delta$ for each $1\leq j\leq k$.
The number $k$ is the {\em length\/} of the path $\pi$ and
it is said that the path $\pi$ starts in state $s_0$ and ends in state~$s_k$.
Moreover, the element $u=x_1x_2\ldots x_k = x_1\circ x_2\circ \ldots \circ x_k\in M$ 
is the {\em label\/} of the path $\pi$.
In particular, the {\em null path\/} is a proper path of the form $(s,\epsilon,s)$, where $s\in S$.
A {\em successful path\/} is a proper path which starts in an initial state and ends in a final state.
%The above path~(\ref{e-path}) is also denoted by
%\begin{eqnarray}
%\pi = (s_0,x_1x_2\ldots x_k,s_k).
%\end{eqnarray}

The {\em generalized transition relation\/} $\Delta^*$ is the smallest subset of $S\times M\times S$ 
containing $\Delta$ with the following closure properties:
\begin{itemize}
\item For each $s\in S$, $(s,\epsilon,s)\in\Delta^*$.
\item For all $s_1,s_2,s_3\in S$ and $u,x\in M$, if $(s_1,u,s_2)\in\Delta^*$ and $(s_2,x,s_3)\in\Delta$, 
then $(s_1,ux,s_3) \in \Delta^*$.
\end{itemize}
Triples $(s,u,s')\in\Delta^*$ are called {\em generalized transitions}.
The generalized transition $(s,u,s')\in\Delta^*$ {\em begins\/} in state $s$, {\em ends\/} in state $s'$ 
and has {\em label\/} $u$.

Let $\cA=(S,M,I,F,\Delta)$ be a monoidal automaton.
The set of labels of all successful paths in $\cA$, i.e.,
\begin{eqnarray}
L(\cA) = \{ u\in M\mid \exists i\in I: \exists f\in F: (i,u,f)\in \Delta^*\},
\end{eqnarray}
is called the {\em language\/} accepted by $\cA$. 
A subset $L$ of $M$ is called a {\em monoidal automaton language\/} over $M$ %or {\em $M$-automaton language\/} 
if there exists a monoidal automaton $\cA$ such that $L=L(\cA)$.

Let $(M,\circ,e)$ be a monoid.
Each subset $L$ of $M$ is a {\em  monoidal language} over $M$. %or {\em $M$language.
The {\em  monoidal regular languages} over $M$ %or {\em regular $M$-languages}
are monoidal languages over $M$ which are inductively defined as follows:
\begin{itemize}
\item $\emptyset$ and $\{m\}$ for each $m\in M$ are monoidal regular languages over $M$.
\item If $L_1$ and $L_2$ are monoidal regular languages over $M$, then $L_1\cup L_2$ (union), 
$L_1\circ L_2$ (monoidal product or concatenation) and $L_1^*$ (monoidal Kleene star) 
are monoidal regular languages over $M$.
\end{itemize}

\begin{proposition}[Regular Languages]\label{p-m-reg}
A monoidal language is regular iff it is a monoidal automaton language
\end{proposition}

Let $(M,\circ,e)$ and $=(M',\odot,e')$ be monoids, let $\phi:M\rightarrow M'$ be a monoid homomorphism,
and let $\cA=(S,M,I,F,\Delta)$ be a monoidal automaton. 
The monoidal automaton 
\begin{eqnarray}
\cA' = (S,M',I,F,\Delta') 
\end{eqnarray}
with the transition relation
\begin{eqnarray}
\Delta' = \{(s,\phi(x),s')\mid (s,x,s') \in\Delta \}
\end{eqnarray}
is the {\em homomorphic image} of $\cA$ under $\phi$.
\begin{proposition}[Homomorphic Images]
If $\cA$ is a monoidal automaton over $M$ and $\cA'$ is its homomorphic image 
under the homomorphism $\phi:M\rightarrow M'$, then
$$L(\cA') = \phi(L(\cA)).$$
\end{proposition}

{\em Classical automata} are monoidal automata 
where the underlying monoid of input symbols is the free monoid $\Sigma^*$ over an alphabet $\Sigma$ 
and the transition labels are in the set $\Sigma\cup\{\epsilon\}$.
%whose transition labels are in the set $\Sigma\cup\{\epsilon\}$.
An monoidal automaton language accepted by a classical automaton is called {\em classical automaton language}.

\begin{proposition}[Classical Languages]
Each monoidal automaton is the homomorphic image of a classical automaton.
Each monoidal automaton language can be established as the homomorphic image of a classical automaton language.
\end{proposition}

%Let $M=(M,\circ,e)$ be a monoid.
%A {\em  monoidal language}  over $M$ is simply a subset of $M$.
%The {\em  monoidal regular languages} over $M$ are inductively defined:
%\begin{itemize}
%\item $\emptyset$ and $\{m\}$ for each $m\in M$ are monoidal regular languages over $M$.
%\item If $L_1,L_2\subseteq M$ are monoidal regular languages over $M$, then $L_1\cup L_2$ (union), 
%$L_1\circ L_2$ (monoidal product) and $L_1^*$ (monoidal Kleene star) are monoidal regular languages.
%\end{itemize}
%
%\begin{proposition}[Regular Languages]
%A monoidal language is regular iff it is a monoidal automaton language
%\end{proposition}

\begin{proposition}[Closure Properties]\label{p-m-cl}
The class of monoidal automaton languages is closed under monoid homomorphisms.
The class of monoidal automaton languages is closed 
under the regular operations union, monoidal product and monoidal Kleene star.
\end{proposition}

\begin{example}\label{e-4}
Consider the monoidal automaton $\cA$ in Fig.~\ref{a-fsa}
which has state set $S=\{i,f\}$, initial state $i$, final state $f$ and transitions
$(i,x,i)$, $(i,y,f)$, $(f,x,f)$ and $(f,y,i)$.
\begin{itemize}
\item
As a classical automaton over the free monoid $\{x,y\}^*$, 
the language is
$$\{x^{j_1}y x^{j_2}y\ldots y x^{j_n}\mid j_1,j_2,\ldots,j_n\geq 0, n\geq 0, \# y \equiv 1\!\!\!\!\mod 2 \}.$$
%This language is regular.
\item
As a monoidal automaton over the commutative monoid $M$ given by the presentation
$\langle x,y\mid xy=yx\rangle$, the language is
$$\{x^{i}y^j\mid i,j\geq 0, j \equiv 1\!\!\!\!\mod 2 \}.$$
%This language is also regular.
%NERODE: 2 classes y odd and y even.
\end{itemize}
\EXX
\end{example}
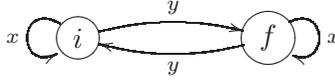
\begin{figure}[hbt]
\begin{center}
\mbox{$
\xymatrix{
%\txt{start}\ar@{-->}[d]_{0.5} && \txt{start}\ar@{-->}[d]_{0.5} \\
*++[o][F-]{i} 
 \ar@(ul,dl)[]_{x}
\ar@/^/[rr]^{y}
%\ar@{-->}[d]
&&
*++[o][F-]{f} 
 \ar@(ur,dr)[]^{x}
%\ar@{-->}[d]
\ar@/^/[ll]^{y} \\
%{0.5(h),0.5(t)}& \txt{output} & {0.75(h),0.25(t)}
}
$}
\end{center}
\caption{Diagram of monoidal automaton.}\label{a-fsa}
\end{figure}

Let $(M_1,\circ_1,e_1)$ and $(M_2,\circ_2,e_2)$ be monoids.
Their Cartesian product $M = M_1\times M_2$ is also monoid, where the associative operation and
the identity element are defined pairwise.

A {\em monoidal $2$-tape automaton} over $M$ %or {\em $2$-tape $M$-automaton} 
is a monoidal automaton $\cA=(S,M,I,F,\Delta)$ 
over a Cartesian product of monoids $M = M_1\times M_2$.
%It is assumed that the generating set $G_M$ is defined component-wise in terms of the generating sets
%$G_{M_1}$  and  $G_{M_2}$, i.e.,
%\begin{eqnarray}
%G_M = \{(x_1,e_2)\mid x_1\in G_{M_1} \} \cup \{(e_1,x_2)\mid x_2\in G_{M_2} \}.
%\end{eqnarray}
A {\em monoidal 2-tape language\/} over $M$ %or {\em 2-tape $M$-language\/} 
is a monoidal language over $M$ accepted by a monoidal $2$-tape automaton over $M$.
These notions can be extended to monoidal $n$-tape automata and monoidal $n$-tape languages for $n\geq 2$.

The class of monoids is closed under Cartesian products and therefore the monoidal $n$-tape automata
are a special case of the monoidal automata.

\begin{proposition}[Inverse Relations]
Let $\cA=(S,M_1\times M_2,I,F,\Delta)$ be a monoidal 2-tape automaton.
Then for the monoidal 2-tape automaton $\cA'=(S,M_2\times M_1,I,F,\Delta')$ with transition relation
$$\Delta' = \{(s,(y,x),s')\mid (s,(x,y),s')\in\Delta\},$$
we have $L(\cA') = L(\cA)^{-1}$, 
where $L(\cA)^{-1} = \{(v,u)\mid (u,v)\in L(\cA)\}$.
\end{proposition}

\begin{proposition}[Projections]
Let $\cA=(S,M_1\times M_2,I,F,\Delta)$ be a monoidal $2$-tape automaton.
Then for the monoidal automaton
$\cA'=(S,M_1,I,F,\Delta')$ with transition relation
$$\Delta' = \{(s,x_1,s')\mid (s,(x_1,x_2),s')\in\Delta\}$$
we have $L(\cA') = L(\cA)_1$, 
where $L(\cA)_1 = \{u_1\mid (u_1,u_2)\in L(\cA)\}$.
\end{proposition}
%This result can be extended to monoidal $n$-tape finite-state automata for $n\geq 2$.

\begin{proposition}[Cartesian Products]
Let $\cA_1=(S_1,M_1,I_1,F_1,\Delta_1)$ and $\cA_2=(S_2,M_2,I_2,F_2,\Delta_2)$ be monoidal automata.
Then for the monoidal 2-tape automaton
$\cA =(S_1\times S_2,M_1\times M_2,I_1\times I_2,F_1\times F_2,\Delta)$ with transition relation
$$\Delta = \{((s_1,s_2),(x_1,x_2),(s'_1,s'_2))\mid (s_i,x_i,s'_i)\in\Delta_i,i=1,2\},$$
we have $L(\cA) = L(\cA_1)\times L(\cA_2)$.
\end{proposition}
%The last two results can be extended to monoidal $n$-tape finite-state automata for $n\geq 2$.

\begin{proposition}[Closure Properties]
The class of monoidal automaton languages is closed under Cartesian products and projections.
The class of monoidal $2$-tape languages is closed under inverse relations.
\end{proposition}

\begin{example}
Consider the two monoidal 2-tape automata over the monoid $M = M_1\times M_2$ 
given in Fig.~\ref{f-gfsa},
where $M_1$ is the commutative monoid given by the presentation $\langle x,y\mid xy=yx\rangle$ 
and $M_2$ is the free monoid $\{z\}^*$.

The language of the first automaton is $$L_1 = \{(x^iy^j,z^i)\mid i\geq 1, j\geq 0\}$$ and 
the language of the second automaton is $$L_2 = \{(x^jy^i,z^i)\mid i\geq 1, j\geq 0\}.$$
Both languages are regular.
The intersection of both languages is
$$L = L_1\cap L_2 = \{(x^iy^i,z^i)\mid i\geq 1\}$$
and thus the projection onto the first component is
$$L' = \{ x^iy^i\mid i\geq 1\}.$$
The language $L'$ is not regular~\cite{salomaa}.
But as already shown, each monoidal automaton language is regular and 
the class of monoidal automaton languages is closed under projection. 
Therefore, the language $L$ cannot be a monoidal automaton language. 
Thus the class of monoidal automaton languages is not closed under intersection.
By De Morgan's law, the class of monoidal automaton languages is also not closed under complement
since by Prop.~\ref{p-m-cl} it is closed under union.
\EXX
\end{example}
\begin{figure}[hbt]
\begin{center}
\mbox{$
\xymatrix{
*++[o][F-]{i}\ar@(ul,dl)[]_{(x,z),(y,\epsilon)} \ar@{->}[rr]^{(x,z)} && *++[o][F-]{f} 
&&& *++[o][F-]{i}\ar@(ul,dl)[]_{(y,z),(x,\epsilon)} \ar@{->}[rr]^{(y,z)} && *++[o][F-]{f} 
}
$}
\end{center}
\caption{Diagrams of monoidal 2-tape automata.}\label{f-gfsa}
\end{figure}
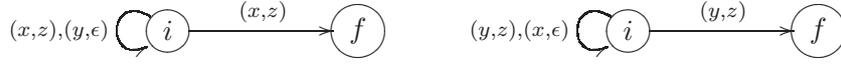

%NOT SUITED, RR_>0 not finitely generated ???
%\begin{example}
%Consider the free monoid $\Sigma^*$ over the alphabet $\Sigma=\{a,b\}$ and
%the monoid of positive real numbers $\RR_{>0}$ under multiplication.
%The monoidal 2-tape finite-state automaton $\cA$ in Fig.~\ref{a-fsa2} 
%defined over the monoid $M = \Sigma^*\times \RR_{> 0}$ 
%has state set $S=\{i,f\}$, initial state $i$, final state $f$ and transitions
%$(i,(a,0.7),i)$, $(i,(b,0.2),i)$ and $(i,(\epsilon,0.1),f)$.
%
%The numbers can be interpreting as the probability of producing the corresponding letter.
%Then each accepted word in $\Sigma^*$ comes with the probability of producing it.
%For instance, the word $u=aab$ is accepted with the probability $0.7\cdot 0.7\cdot 0.2\cdot 0.1$.
%\EXX
%\end{example}
%\begin{figure}[hbt]
%\begin{center}
%\mbox{$
%\xymatrix{
%%\txt{start}\ar@{-->}[d]_{0.5} && \txt{start}\ar@{-->}[d]_{0.5} \\
%*++[o][F-]{i} 
% \ar@(ul,dl)[]_{(a,0.7),(b,0.2)}
%%\ar@/^/[rr]^{y}
%\ar@{->}[rr]^{(\epsilon,0.1)}
%&&
%*++[o][F-]{f} 
%% \ar@(ur,dr)[]^{x}
%%\ar@{-->}[d]
%%\ar@/^/[ll]^{y} \\
%%{0.5(h),0.5(t)}& \txt{output} & {0.75(h),0.25(t)}
%}
%$}
%\end{center}
%\caption{Diagram of 2-tape automaton.}\label{a-fsa2}
%\end{figure}

%%%%%%%%%%%%%%%%%%%%%%%%%%%%%%%%%%%%%%%%%%%%%%%%%%%%%%%%%%%%%%%%%%%%
\section{Monoidal Generalized Automata}\label{s-5}

%The unification
Generalized finite-state automata and monoidal finite-state automata can
be unified to monoidal generalized finite-state automata. 
%by replacing the free monoid by an arbitrary (finitely generated) monoid.
%This generalization has been applied to non-deterministic finite-state automata~\cite{mihov}.

A {\em monoidal generalized automaton\/} %over a monoid $M$ %or {\em generalized $M$-automaton} 
is a quintuple 
$$\cA = (S,M, \{Q(x)\mid x\in G_M\}, \pi,f),$$ 
where 
\begin{itemize}
\item $S$ is the non-empty finite set of {\em states},
\item $(M,\circ,e)$ is a finitely generated monoid, where $M$ is the set of {\em input symbols},
\item $Q$ is a finite collection of $n\times n$ matrices $Q(x)$ with $x\in G_M$, 
where $n$ is the number of states and $G_M$ is a generating set of $M$,
\item $\pi\in\RR^n$ is the {\em initial vector}\index{initial distribution} written as row vector, and
\item $f\in\RR^n$ is the {\em final vector}\index{final state vector} written as column vector.
\item {\em Extension postulate:} 
The mapping $Q:G_M\rightarrow\RR^{n\times n}$ 
can be uniquely extended to a monoid homomorphism $Q:M\rightarrow \RR^{n\times n}$ 
such that for each word $u=x_1\ldots x_k\in M$,
\begin{eqnarray}\label{e-post0}
Q(u) = Q(x_1) \cdots Q(x_k) 
\end{eqnarray}
and particularly $Q(\epsilon) = I_n$.
\end{itemize}
%Note that if there is such an extension to a monoid homomorphism $Q:M\rightarrow \RR^{n\times n}$, 
%it is uniquely determined.
%Indeed, suppose  $R:M\rightarrow \RR^{n\times n}$ is another monoid homomorphism exten\-ding 
%$Q:G_M\rightarrow \RR^{n\times n}$.
%Then $Q(x)=R(x)$ for all $x\in G_M$.
%By induction, if $u=vx\in M$ where $v\in M$ and $x\in G_M$, then
%$Q(u) = Q(v)Q(x) = R(v)R(x) =R(u)$.

The extension postulate ensures that the monoid operation is compa\-tible with the matrix multiplication.
In view of generalized or stochastic automata, 
this postulate is a direct consequence of the universal property of free monoids~\cite{cliff,mihov}.
%It is required here as it guarantees well-definedness, 
%since the elements $u\in M$ may not have a unique representation in terms of the generators.

\begin{example}\label{e-xy0}
Consider the commutative monoid $M$ given by the presentation
$\langle x,y\mid xy=yx\rangle$.
Each element of $M$ has the form $x^iy^j$ for some $i,j\geq 0$.

Define the matrices
$$Q(x) = \left(\begin{array}{rr} 1 & 1 \\ 0 & 1 \end{array}\right)
\quad\mbox{and}\quad
Q(y) = \left(\begin{array}{rr} 1 & -1 \\ 0 & 1 \end{array}\right).$$
The mapping $Q:\{x,y\}\rightarrow\RR^{2\times 2}$ extends to a unique monoid-homomorphism
$Q:M\rightarrow \RR^{2\times 2}$, where 
$$Q(x^iy^j) = Q(x)^iQ(y)^j = 
\left(\begin{array}{cc} 1 & i-j \\ 0 & 1 \end{array}\right),\quad i,j\geq 0.$$
\EXX
\end{example}

\begin{example}\label{e-xy00}
Consider the commutative monoid $M$ given by the presentation
$\langle x,y\mid xy=yx\rangle$.
Each element of $M$ has the form $x^iy^j$ for some $i,j\geq 0$.

Define the matrices
$$Q(x) = \left(\begin{array}{rr} 1 & 1 \\ 0 & 1 \end{array}\right)
\quad\mbox{and}\quad
Q(y) = \left(\begin{array}{rr} 1 & 0 \\ 1 & 1 \end{array}\right).$$
The mapping $Q:\{x,y\}\rightarrow\RR^{2\times 2}$ cannot be extended to a monoid-homo\-mor\-phism,
since $xy=yx$\/ but the matrices $Q(x)$ and $Q(y)$ do not commute, 
$$Q(x)Q(y) = \left(\begin{array}{rr} 2 & 1 \\ 1 & 1 \end{array}\right)
\quad\mbox{and}\quad
Q(y)Q(x) = \left(\begin{array}{rr} 1 & 1 \\ 1 & 2 \end{array}\right).$$
\EXX
\end{example}
%The submonoid of $M$ generated by the set $G_M$ is finitely generated and is subsequently denoted by $\cU_M$. 
%In the following, it is assumed that $G_M$ is a generating set of $M$ and so $M=U_M$.

Let $\cA = (S,M, \{Q(x)\mid x\in G_M\}, \pi,f)$ be a monoidal generalized automaton and 
$\lambda$ be a real number.
The set
\begin{eqnarray}
L(\cA,\lambda) =\{u\in M\mid \pi Q(u) f>\lambda\}
\end{eqnarray}
is the {\em language\/} accepted by $\cA$ w.r.t.\ the {\em cut point\/} $\lambda$.
A subset $L$ of $M$ is called a {\em monoidal generalized automaton language\/} 
if there exists a monoidal generalized automaton $\cA$
and a real number $\lambda$ such that $L=L(\cA,\lambda)$.
%A cut point $\lambda$ is {\em isolated\/} for the monoidal generalized automaton $A$ over the monoid $M=(M,\circ,e)$
%if there is a real number $\delta>0$ such that for all $x\in U_M$,
%$$|\lambda - \pi P(x) f| \geq \delta.$$

{\em Classical generalized automata\/} are monoidal generalized automata 
where the underlying monoid is the free monoid $\Sigma^*$ over the alphabet $G_M=\Sigma$.
In this case, the extension postulate follows directly from the universal property of free monoids.
A monoidal generalized automaton language accepted by a classical generalized automaton 
is called {\em classical generalized automaton language}.

\begin{proposition}\label{p-malg}
Every monoidal automaton language is a monoidal generalized automa\-ton language.
\end{proposition}
\begin{proof}
Let $L$ be a monoidal automaton language.
Then there is a monoidal automaton $\cA=(S,M,I,F,\Delta)$ such that $L=L(\cA)$.

Let $S=\{s_1,\ldots,s_n\}$ 
and let $G_M=\{x\mid (s,x,s')\in \Delta\}$ be the set of transition labels of $\cA$.
We may assume that $G_M$ is a generating set of $M$.

Define the monoidal generalized automaton 
$$\cA'=(S,M,\{Q(x)\mid x\in G_M\},\pi,f),$$
where $Q(x) = (q_{i,j}^{(x)})$ is the $n\times n$ matrix with entries 
$$q_{i,j}^{(x)}
= \left\{ \begin{array}{ll}
1 & \mbox{if } (s_i,x,s_j)\in \Delta,\\
0 & \mbox{otherwise,}
\end{array} \right.$$
$\pi=(\pi_1,\ldots,\pi_n)$, where $\pi_i=1$ if $s_i\in I$ and $\pi_i=0$ otherwise,
and $f=(f_1,\ldots,f_n)^T$, where $f_i=1$ if $s_i\in F$ and $f_i=0$ otherwise.

Let $u=x_1\ldots x_k\in M$.
Then by the property of matrix multiplication, 
the matrix $Q(u) = (q_{ij}^{(u)}) = Q(x_1)\cdots Q(x_k)$ has entry $q_{ij}^{(u)}>0$ iff $(s_i,u,s_j)\in \Delta^*$.
Thus the extension postulate is satisfied.
In particular,
$u\in L$ iff $(s_i,u,s_j)\in \Delta^*$ for some $s_i\in I$ and $s_j\in F$.
This is equivalent to the condition $\pi Q(u) f>0$.
Hence, $L=L(\cA',0)$.
\end{proof}
In view of Prop.~\ref{p-m-reg}, one obtains the following consequence.
\begin{corollary}
The regular monoidal languages are monoidal generalized automaton languages.
\end{corollary}

%First, we provide a few standardization results.
%Hilfssatz 29
\begin{proposition}\label{p-l}
Let $\cA=(S,M,\{Q(x)\mid x\in G_M\},\pi,f)$ be a monoidal generalized automaton and 
let $\lambda>0$ be a real number.
Then for each real number $\lambda'>0$, there exists a monoidal generalized automaton $\cA'$ such that 
$$L(\cA,\lambda)=L(\cA',\lambda').$$
\end{proposition}
\begin{proof}
Let $S=\{s_1,\ldots,s_n\}$ and put $\alpha=1 - \frac{\lambda'}{\lambda}$. 
Consider the monoidal generalized automaton
$$\cA'=(S',M,\{Q'(x)\mid x\in G_M\},\pi',f'),$$
where $S'=S$, $Q'(x) =Q(x)$ for each $x\in G_M$, 
$\pi'=(1-\alpha)\cdot \pi$ and $f'=f$.
Then for each word $u=x_1\ldots x_k\in M$,
$$\pi'Q'(u)f' = (1-\alpha)\pi Q(u)f = \frac{\lambda'}{\lambda} \pi Q(u) f.$$
Thus 
$\pi'Q'(u)f'>\lambda'$ 
iff
%$\frac{\lambda'}{\lambda} \pi Q(u) f > \lambda'$ iff
$\pi Q(u) f>\lambda$ 
and hence
$L(\cA,\lambda)=L(\cA',\lambda')$.
%Define $S' = S\cup \{s_{n+1}\}$ and
%$$Q'(x) = \left(\begin{array}{c|c}  & 0 \\ Q(x) & \vdots \\     & 0 \\\hline 0\ldots0 & 1 \end{array}\right),\quad x\in G_M.$$
%Put $\alpha=1 - \frac{\lambda'}{\lambda}$ 
%and consider the monoidal generalized finite-state automaton
%$\cA'=(S',M,\{Q'(x)\mid x\in G_M\},\pi',f')$,
%where $\pi'=(1-\alpha)\cdot \pi$ and $f'={f\choose 0}$.
%Then for each word $u=x_1\ldots x_k\in M$,
%$$\pi'Q'(u)f' = (1-\alpha)\pi Q(u)f = \frac{\lambda'}{\lambda} \pi Q(u) f.$$
%Thus $\pi Q(u) f>\lambda$ iff $\pi'Q'(u)f'>\lambda'$ and hence
%$L(\cA,\lambda)=L(\cA',\lambda')$.
\end{proof}

%%Hilfssatz 30
%\begin{proposition}\label{p-z}
%Let $\cA=(S,M,\{Q(x)\mid x\in G_M\},\pi,f)$ be a monoidal generalized finite-state automaton and $\lambda < 1$.
%Then there exists a monoidal generalized finite-state automaton $\cA'$ 
%in which the initial vector $\pi'$ is a unit vector, such that $L(\cA,\lambda)=L(\cA',\lambda)$.
%\end{proposition}
%\begin{proof}
%Let $S=\{s_1,\ldots,s_n\}$.
%Define $S' = S\cup \{s_{0}\}$ and
%$$Q'(x) = \left(\begin{array}{c|c} 0 & \pi Q(x) \\\hline 0    &  \\ \vdots & Q(x)\\ 0 &  \end{array}\right),\quad x\in G_M.$$
%Consider the monoidal generalized finite-state automaton
%$\cA'=(S',M,\{Q'(x)\mid x\in G_M\},\pi',f')$,
%where $\pi'=(1,0,\ldots,0)$, and
%$f' = {0\choose f}$ if $\pi f\leq \lambda$ and
%$f' = {1\choose f}$ if $\pi f> \lambda$.
%
%Then $\pi Q(\epsilon) f = \pi f>\lambda$ iff $\pi'Q'(\epsilon)f' = \pi'f'=1>\lambda$ 
%and so $\epsilon\in L(\cA,\lambda)$ iff $\epsilon\in L(\cA',\lambda)$.
%Moreover, for each word  $u=x_1\ldots u_k\in M$,
%$$\pi'Q'(u)f' = \pi Q(u)f.$$ 
%Thus $\pi Q(u) f>\lambda$ iff $\pi'Q'(u)f'>\lambda$ and hence
%$L(\cA,\lambda)=L(\cA',\lambda)$.
%\end{proof}

\section{Turakainen's Result}
Turakainen's result (Thm.~\ref{p-tura0}) and the matrix characterization of stochastic languages (Thm.~\ref{p-claus})
will be considered in the monoidal setting.
In view of Turakainen's theorem, let 
$$\cA=(S,M,\{Q(x)\mid x\in G_M\},\pi,f)$$ 
be a monoidal generalized automaton,
$\lambda$ a real number and $L=L(\cA,\lambda)$ the language accepted by $\cA$ 
w.r.t.\ the cut point $\lambda$.
In the following, let $S=\{s_1,\ldots,s_n\}$.
The proof of the generalization of Turakainen's result (Thm.~\ref{t-tura-gen}) 
will be broken down into several steps and will be conducted for the language 
$L' = L\setminus\{\epsilon\}$.
The empty word will be considered at the end.

Note that in the following assertions, 
a new automaton $\cA'$ will always be constructed from a given one, $\cA$, satisfying the extension postulate 
such that the monoid homomorphism $Q':M\rightarrow\RR^{n\times n}$ in $\cA'$  
is an extension or modification of the monoid homomorphism $Q:M\rightarrow\RR^{n\times n}$ in $\cA$.  

%Turakainen, proof part I accepted by $\cA$ 
\begin{proposition}\label{p-cr0}
There exists a monoidal generalized automaton
$$\cA_1=(S_1,M,\{Q_1(x)\mid x\in G_M\},\pi_1,f_1),$$ 
whose matrices $Q_1(x)$ have column and row sums equal to~0, such that
$$L'=L(\cA_1,\lambda)\setminus\{\epsilon\}.$$
%Let $\cA=(S,M,\{Q(x)\mid x\in G_M\},\pi,f)$ be a monoidal generalized finite-state automaton and 
%let $\lambda$ be a real number.
%Then there exists a monoidal generalized finite-state automaton $\cA'$, 
%where the matrices have column and row sums equal to~0, such that 
%$$L(\cA,\lambda)=L(\cA',\lambda).$$
\end{proposition}
\begin{proof}
Define $S_1 = S\cup \{s_0,s_{n+1}\}$, $\pi_1=(0,\pi,0)$, 
$f_1=\left(\begin{array}{c} 0\\ f\\ 0 \end{array}\right)$ and
$$Q_1(x) = \left(\begin{array}{c|ccc|c} 
0            & 0&\ldots& 0 & 0 \\ \hline
-\sigma_1(x) &  &      &   & 0  \\
\vdots       && Q(x) &     &\vdots\\
-\sigma_n(x) &  &      &   & 0  \\\hline
\sigma''(x) & -\sigma'_1(x)&\ldots &-\sigma'_n(x) & 0  \\
\end{array}\right),\quad x\in G_M,$$
where $\sigma_i(x)$ is the $i$th row sum of $Q(x)$, $\sigma'_j(x)$ is the $j$th column sum of $Q(x)$, 
and $\sigma''(x)$ is the sum of all entries of $Q(x)$.

For each non-empty word $u=x_1\ldots x_k\in M$,
the matrix $Q_1(u) = Q_1(x_1)\cdots Q_1(x_k)$ has the same form as that of the generators,
i.e., the column and row sums of $Q_1(u)$ are equal to~0.
It is easy to check that
$$\pi_1Q_1(u)f_1 = \pi Q(u)f.$$ 
Thus $\pi Q(u) f>\lambda$ iff $\pi_1Q_1(u)f_1>\lambda$ and hence
the result follows.
%$L(\cA,\lambda)=L(\cA_1,\lambda)$.
\end{proof}

%Turakainen, proof part II
\begin{proposition}\label{p-zz}
There exists a monoidal generalized automaton
$$\cA_2=(S_2,M,\{Q_2(x)\mid x\in G_M\},\pi_2,f_2),$$
whose matrices $Q_2(x)$ are non-negative, such that
$$L'=L(\cA_2,\lambda)\setminus\{\epsilon\}.$$
%$$L(\cA,\lambda)\setminus\{\epsilon\}=L(\cA',\lambda)\setminus\{\epsilon\}.$$
%Let $\cA=(S,M,\{Q(x)\mid x\in G_M\},\pi,f)$ be a monoidal generalized automaton,
%where the matrices have column and row sums equal to~0, 
%and let $\lambda$ be a real number.
%Then there exists a monoidal generalized automaton $\cA'$, 
%where matrices are non-negative, such that
%$$L(\cA,\lambda)\setminus\{\epsilon\}=L(\cA',\lambda)\setminus\{\epsilon\}.$$
\end{proposition}
\begin{proof}
%Reconsider the monoidal generalized automaton
%$\cA'=(S',M,\{Q'(x)\mid x\in G_M\},\pi',f')$ 
%in the proof of Prop.~\ref{p-cr0}.
Let $S_2=S_1\cup\{s_{n+2}\}$ and $m=n+2$.
Define
$$\pi_2=\left(\pi_1,\frac{\alpha}{m}\right)
\quad\mbox{and} \quad 
f_2=\left(\begin{array}{c} f_1\\-1\end{array}\right),$$
where $\alpha$ is the product of the sum of the components of $\pi_1$ 
and the sum of the components of $f_1$, i.e.,
$$\alpha = \left(\sum_{i=0}^{m-1} \pi_{1i} \right)\cdot \left(\sum_{i=0}^{m-1} f_{1i} \right).$$

%Define $S'' = S'\cup \{s_{n+2}\}$.
%Put $\pi''=(\pi',\alpha/(n+2))$ and $f'=\left(\begin{array}{c} f'\\-1\end{array}\right)$,
%where $\alpha$ is the product of the sum of the components of $\pi'$ and the sum of the components of $f'$,
%$$\alpha = \left(\sum_{i=0}^{n+1} \pi'_i \right)\cdot \left(\sum_{i=0}^{n+1} f'_i \right)
%= \left(\sum_{i=1}^{n} \pi_i \right)\cdot \left(\sum_{i=1}^{n} f_i \right).$$

Let $r$ be a real number.
Write $B_r$ for the $m\times m$ matrix whose entries are~$r$, i.e.,
$$B_r = \left(\begin{array}{ccc} r & \ldots & r\\ \vdots  & &  \vdots \\ r & \ldots & r \end{array}\right).$$

Choose $r\geq 0$ such that for each $x\in G_M$, the matrix $Q_1(x)+B_r$ is non-negative.
%where $Q(x)$ is the matrix in the proof of Prop.~\ref{p-cr0}.
Since the column and row sums of the matrices $Q_1(x)$, $x\in G_M$, are zero, 
the matrices $Q_1(x)\cdot B_r$ and $B_r\cdot Q_1(x)$ are both zero matrices.
It follows that for all $x,y\in G_M$,
$$(Q_1(x)+B_r)\cdot (Q_1(y)+B_r) = Q_1(x)Q(y) + B_rB_r = Q_1(xy) + B_{mr^2}.$$
Moreover, for each nonempty word $u=x_1\ldots x_k\in M$, 
the column and row sums of the matrix $Q_1(u)$ are zero as well.
Thus for all non-empty words $u,v\in M$,
$$(Q_1(u)+B_r)\cdot (Q_1(v)+B_r) = Q_1(uv) + B_{mr^2}.$$

Define the matrices  
$$Q_2(x) = \left(\begin{array}{cc} Q_1(x)+B_r  & 0 \\ 0 & mr \end{array}\right),\quad x\in G_M.$$
Then for each word $u=x_1\ldots x_k\in M$ of length $k\geq 1$,
$$Q_2(u) = Q_2(x_1)\cdots Q_2(x_k) 
=  \left(\begin{array}{cc} Q_1(u)+B_{m^{k-1}r^k}  & 0 \\ 0 & m^kr^k \end{array}\right).$$
Thus %for each non-empty word $u=x_1\ldots x_k\in M$, %$u\ne\epsilon$, 
\begin{eqnarray*}
\pi_2 Q_2(u) f_2
&=& \pi_1 Q_1(u) f_1 + \pi_1 B_{m^{k-1}r^k}f_1 - \alpha m^{k-1}r^k \\
&=& \pi_1 Q_1(u) f_1,
\end{eqnarray*}
since $\pi_1 B_r f_1 = \sum_{i=1}^{m}\sum_{j=1}^{m} \pi_{1i} r f_{1j} = \alpha r$ 
for each real number $r$.
Hence, by Prop.~\ref{p-cr0}, the result follows.
%Moreover, in view of the empty word, 
%$\pi Q(\epsilon)f = \pi f$ and $\pi'Q'(\epsilon) f' = \pi f - \frac{\alpha}{n}$.
\end{proof}

%Turakainen: proof part III.
\begin{proposition}\label{p-0}
There exists a monoidal generalized automaton
$$\cA_3=(S_3,M,\{Q_3(x)\mid x\in G_M\},\pi_3,f_3),$$
whose matrices $Q_3(x)$ are stochastic, such that
$$L'=L(\cA_3,0)\setminus\{\epsilon\}.$$
%$\cA=(S,M,\{Q(x)\mid x\in G_M\},\pi,f)$ be a monoidal generalized finite-state 
%automaton, where the matrices are non-negative,  and let $\lambda$ be a real number.
%Then there exists a monoidal generalized finite-state automaton $\cA'$,
%where the matrices are stochastic, such that 
%$$L(\cA,\lambda)=L(\cA',0).$$
\end{proposition}
\begin{proof}
Define $S_3 = S_2\cup \{s_{n+3},s_{n+4}\}$,
$\pi_3=(\pi_2,\lambda,0)$, 
$f_3= \left(\begin{array}{c} f_2\\-1\\0\end{array}\right)$ and
$$Q_3(x) 
= \left(\begin{array}{c|cc} 
& 0 & \beta_0(x)\\ 
\frac{1}{\beta} Q_2(x)  & \vdots & \vdots \\ 
& 0 & \beta_{n+2}(x) \\ \hline 
0\ldots 0 & \frac{1}{\beta} &  1-\frac{1}{\beta}\\
0\ldots 0 & 0 &  1
\end{array}\right),\quad x\in G_M,$$
where $\beta>0$ is a real number which is larger 
than the row sums of each matrix $Q_2(x)$, $x\in G_M$,
and $\beta_0(x),\ldots,\beta_{n+2}(x)\geq 0$ 
are real numbers such that the matrix $Q_3(x)$, $x\in G_M$, becomes stochastic.
Then for each word $u=x_1\ldots x_k\in M$ of length $k\geq 1$, 
the matrix $Q_3(u)$ is given by
$$Q_3(u) 
= Q_3(x_1)\cdots Q_3(x_k) 
= \left(\begin{array}{c|cc} 
& 0 & \beta'_0(u)\\ 
\frac{1}{\beta^k} Q_2(u)  & \vdots & \vdots \\ 
& 0 & \beta'_{n+2}(u) \\ \hline 
0\ldots 0 & \frac{1}{\beta^k} &  1-\frac{1}{\beta^k}\\
0\ldots 0 & 0 &  1
\end{array}\right).$$
This matrix is also stochastic; 
the numbers $\beta'_0(u),\ldots,\beta'_{n+2}(u)\geq 0$ are not of interest in the sequel.

For each non-empty word $u=x_1\ldots x_k\in M$, %$x\ne\epsilon$,
$$\pi_3Q_3(u)f_3 
= \frac{1}{\beta^k}\pi_2 Q_2(u)f_2 - \frac{1}{\beta^k}\lambda 
= \frac{1}{\beta^k} \left(\pi_2 Q_2(u) f_2 - \lambda\right).$$ 
Thus $\pi_3 Q_3(u) f_3> 0$ iff $\pi_2 Q_2(u)f_2>\lambda$.
Hence, by Prop.~\ref{p-zz}, the result follows.
%See proof page 122, sentence before IV. It holds also for epsilon.
%Moreover, in view of the empty word, 
%$\pi_3 f_3  =\pi_2 f_2-\lambda >0$ iff $\pi_2 f_2>\lambda$.
%Hence, $L(\cA_3,0)=L(\cA_2,\lambda)$.
\end{proof}

%Turakainen, proof, part IV
\begin{proposition}\label{p-4}
There exists a monoidal generalized automaton 
$$\cA_4 = (S_4,M, \{Q_4(x)\mid x\in G_M\}, \pi_4,f_4),$$ 
where the matrices $Q_4(x)$ are stochastic,
the initial vector $\pi_4$ is a state distribution 
and the final vector $f_4$ is positive, such that for some real number $\lambda'>0$,
$$L'= L(\cA_4,\lambda')\setminus\{\epsilon\}.$$
%be a monoidal generalized finite-state automaton, 
%where the matrices are stochastic.
%Then there exists a monoidal generalized finite-state automaton $\cA'$, 
%where the matrices are stochastic, the initial vector is a state distribution 
%and the final vector is positive, such that
%$$L(\cA,0) = L(\cA',\lambda)$$
%for some suitable real number $\lambda>0$.
\end{proposition}
\begin{proof}
Let $m=n+5$. 
Put $S_4=S_3\cup \{s_{m},\ldots,s_{2m-1}\}$ and 
$$Q_4(x) = \left(\begin{array}{cc} Q_3(x) & 0 \\ 0 & Q_3(x) \end{array}\right),\quad x\in G_M.$$
Then for each word $u=x_1\ldots x_k\in M$,
$$Q_4(u) = Q_4(x_1)\cdots Q_4(x_k) = \left(\begin{array}{cc} Q_3(u) & 0 \\ 0 & Q_3(u) \end{array}\right).$$
The matrices $Q_4(u)$ are also stochastic.

Let $1_m$ be the all-1 row vector of length $m$.
Choose a real number $r>0$ large enough such that the row vector 
$\pi_3 + r\cdot 1_m$ has only positive components.
Then put $$\pi_4 = \frac{1}{R}(\pi_3 + r\cdot 1_m, r\cdot 1_m),$$
where $R$ is taken such that the row sum of $\pi_4$ is equal to~1;
i.e., $\pi_4$ is a state distribution of $S_4$.

Moreover, choose a real number $t>0$ large enough such that the column vector 
$$f_4 = {f_3 + t\cdot 1_m^T\choose -f_3 + t\cdot 1_m^T}$$ has only positive components.

Then for each non-empty word $u\in M$,
\begin{eqnarray*}
\pi_4Q_4(u) f_4 &=& \pi_4Q_4(u) {f_3 \choose -f_3}  + \pi_4Q_4(u) {t\cdot 1_m^T\choose t\cdot 1_m^T}\\
&=& \frac{1}{R}\pi_3 Q_3(u) f_3  + t,
\end{eqnarray*}
since by construction the vector $\pi_4Q_4(u)$ has row sum equal to~1.
Put $\lambda'=t$.
Then $L(\cA_3,0)\setminus\{\epsilon\} = L(\cA_4,\lambda')\setminus\{\epsilon\}$ 
and hence the result follows from Prop.~\ref{p-0}.
\end{proof}

A monoidal generalized automaton $\cA = (S,M, \{P(x)\mid x\in G_M\}, \pi,f)$ 
is a {\em monoidal stochastic automaton\/}
if the matrices $P(x)$, $x\in G_M$, are stochastic,
the initial vector $\pi$ is a probability distribution of the state set,
and the final vector $f$ is a binary vector.
A {\em monoidal stochastic automaton language\/} is 
a monoidal generalized automaton language $L$, which is accepted by
a monoidal stochastic automaton  $\cA$, i.e., 
$L=L(\cA,\lambda)$ for some cut point $\lambda$ with $0\leq\lambda\leq 1$.

%Turakainen, proof, part V
\begin{proposition}\label{p-5}
There exists a monoidal stochastic automaton 
$$\cA_5 = (S_5,M, \{P(x)\mid x\in G_M\}, \pi_5,f_5)$$ 
such that for some $0 \leq \lambda''\leq 1$,
$$L' = L(\cA_5,\lambda'')\setminus\{\epsilon\}.$$
%Let $\cA = (S,M, \{Q(x)\mid x\in G_M\}, \pi,f)$ 
%be a monoidal generalized finite-state automaton, 
%where the matrices are stochastic, the initial vector is a state distribution
%and the final vector is positive.
%Then there exists a monoidal stochastic finite-state automaton $\cA'$ such that
%$$L(\cA,\lambda)\setminus\{\epsilon\} = L(\cA',\lambda')\setminus\{\epsilon\},$$
%where $\lambda'$ depends on $\lambda>0$ and $0 \leq \lambda'\leq 1$.
\end{proposition}
\begin{proof}
The automaton $\cA_4$ has $k=2n+10$ states.
The new automaton $\cA_5$ has $k$-fold many states, 
$$S_5 = \{s_0,\ldots,s_{k^2-1}\}.$$
The final state vector $f_5$ is the binary vector defined by the final state set 
$$F_5 = \{s_0,s_{k+1}, s_{2k+2},\ldots,s_{k^2-1}\} = \{s_{ik+i}\mid 0\leq i\leq k-1\}$$
and the initial state vector is
$$\pi_5 = \frac{1}{k}(\pi_4,\pi_4,\ldots,\pi_4).$$
It is clear that $\pi_5$ is a state distribution of $S_5$.

Put $\alpha = \sum_{i=1}^k f_{4i}>0$ and 
$\alpha_i = \frac{f_{4i}}{\alpha}>0$ for each $1\leq i\leq k$.
Define the $k\times k$ block matrices
$$P(x) = \left(\begin{array}{cccc} 
\alpha_1 Q_4(x) & \alpha_2 Q_4(x) & \ldots & \alpha_k Q_4(x) \\ 
\alpha_1 Q_4(x) & \alpha_2 Q_4(x) & \ldots & \alpha_k Q_4(x) \\ 
\vdots        & \vdots        &        & \vdots \\
\alpha_1 Q_4(x) & \alpha_2 Q_4(x) & \ldots & \alpha_k Q_4(x) 
\end{array}\right),\quad x\in G_M.$$
Since $\sum_{i=1}^k\alpha_i=1$, for each non-empty word $u=x_1\ldots x_k\in M$,
$$P(u) = P(x_1)\cdots P(x_k) = \left(\begin{array}{cccc} 
\alpha_1 Q_4(u) & \alpha_2 Q_4(u) & \ldots & \alpha_k Q_4(u) \\ 
\alpha_1 Q_4(u) & \alpha_2 Q_4(u) & \ldots & \alpha_k Q_4(u) \\ 
\vdots        & \vdots        &        & \vdots \\
\alpha_1 Q_4(u) & \alpha_2 Q_4(u) & \ldots & \alpha_k Q_4(u) 
\end{array}\right).$$
Then 
\begin{eqnarray*}
\pi_5 P(u) 
&=&  \frac{1}{k}\left(k\alpha_1\pi_4 Q_4(u),\ldots,k\alpha_k\pi_4 Q_4(u)\right)\\
&=&  \left(\alpha_1\pi_4 Q_4(u),\ldots,\alpha_k\pi_4 Q_4(u)\right).
\end{eqnarray*}
Thus by the choice of the final vector,
\begin{eqnarray*}
\pi_5 P(u) f_5
&=&  \pi_4 Q_4(u) \left(\begin{array}{c} \alpha_1\\\vdots\\\alpha_k \end{array}\right)
= \frac{1}{\alpha}\pi_4 Q_4(u) f_4.
\end{eqnarray*}
Put $\lambda'' = \lambda'/\alpha$. 
Then $\pi_5 P(u) f_5>\lambda''$ iff $\pi_4 Q_4(u) f_4>\lambda'$.
Therefore by Prop.~\ref{p-4}, 
$L'=L(\cA_4,\lambda') \setminus\{\epsilon\} =L(\cA_5,\lambda'') \setminus\{\epsilon\}$.
Finally, by the choice of $\lambda'$ in the proof of Prop.~\ref{p-4}, 
we have $0< \lambda'\leq \alpha$ and so $0\leq \lambda''\leq 1$.
\end{proof}

%Turakainen: proof part VI.
%OWN PROOF
\begin{proposition}\label{p-0e}
Let $L$ be a monoidal stochastic automaton language and let 
$$\cA=(S,M,\{P(x)\mid x\in G_M\},\pi,f)$$ be a monoidal stochastic automaton 
such that $L(\cA,\lambda) = L\setminus\{\epsilon\}$ for some cut point $0\leq\lambda\leq 1$.
Then there exists a monoidal stochastic automaton $\cA'$ 
such that for some $0\leq \lambda'\leq 1$,
$$L=L(\cA',\lambda').$$ 
\end{proposition}
\begin{proof}
First, suppose $\epsilon\not\in L$.  
Then put $\cA'=\cA$.

Second, suppose $\epsilon\in L$.  
First, assume that $\lambda=1$.
Then $L(\cA,\lambda)=\emptyset$ and so $L=\{\epsilon\}$.
Define the monoidal stochastic automaton 
$$\cA'=(\{s_1,s_2\},M,\{P'(x)\mid x\in G_M\},(1,0),(1,0)^T),$$
where
$$ P'(x) = \left(\begin{array}{cc} 0&1 \\ 0& 1 \end{array}\right), \quad x\in G_M.$$
Then for each non-empty word $u=x_1\ldots x_k\in M$,
$$ P'(u) = \left(\begin{array}{cc} 0&1 \\ 0& 1 \end{array}\right).$$
Therefore, $(1,0)P'(u)(1,0)^T = 0$.
Moreover, $(1,0)(1,0)^T=1$.
Thus for any cut point $0\leq \lambda'<1$, we obtain
$L(\cA',\lambda') = \{\epsilon\}$.

Finally, suppose $\epsilon\in L$ and $\lambda <1$. 
Let $S=\{s_1,\ldots,s_n\}$.
Define $S' = S\cup\{s_{n+1}\}$ and
$$P'(x) 
= \left(\begin{array}{c|c} & 0 \\ P(x) & \vdots \\  & 0\\\hline \pi P(x) & 0 \end{array}\right),
\quad x\in G_M.$$
Then for each non-empty word $u=x_1\ldots x_k\in M$,
$$P'(u) 
=P'(x_1)\cdots P'(x_k)
= \left(\begin{array}{c|c} & 0 \\ P(u) & \vdots \\  & 0\\\hline \pi P(u) & 0 \end{array}\right),
\quad x\in G_M.$$
Consider the monoidal stochastic automaton
$$\cA'=(S',M,\{P'(x)\mid x\in G_M\},\pi',f'),$$ 
where $\pi' = (0,\ldots,0,1)$ and $f' = {f\choose 1}$.
Then for each non-empty word $u\in M$,
$$\pi'P'(u)f' = \pi P(u) f$$
and so 
$\pi'P'(u)f' >\lambda $ iff $\pi P(u) f>\lambda$.
Hence, $L(\cA',\lambda) \setminus\{\epsilon\} = L(\cA,\lambda) \setminus\{\epsilon\}$.
Moreover, in view of the empty word,
$\pi'f' = 1>\lambda$ and so $\epsilon \in L(\cA',\lambda)$.
\end{proof}

In view of the previous results, Turakainen's theorem holds in the more general monoidal setting.
%Turakainen, Satz 21
\begin{theorem}[Generalized Turakainen]\label{t-tura-gen}
Every monoidal generalized automaton language is a monoidal stochastic automaton language.
\end{theorem}
The construction in the proof reveals the following.
\begin{corollary}
Let $L$ be a monoidal generalized automaton language accepted by a monoidal generalized automaton with $n$ states.
Then there exists a monoidal stochastic automaton with at most $(2n+10)^2 + 1$ states that accepts $L$.  
\end{corollary}

%The following matrix characterization generalizes the matrix characterization of
%stochastic automaton languages (Prop.~\ref{p-claus}).
\begin{theorem}[Matrix Characterization]\label{p-mga-L}
Let\/ $M=(M,\circ,e)$ be a finitely generated mono\-id.
A subset $L$ of\/ $M$ is a monoidal stochastic automaton language 
iff there exists a finite collection of $n\times n$ matrices
$\{Q(x)\mid x\in G_M\}$ for some integer $n\geq 1$, where $G_M$ is a generating set of\/ $M$, 
such that the mapping $Q:G_M\rightarrow\RR^{n\times n}$ extends to a unique monoid homomorphism
$Q:M\rightarrow\RR^{n\times n}$ as in~(\ref{e-post0})
and then for each non-empty word $u\in M$, %$u\ne \epsilon$,
$$u\in L\quad\Longleftrightarrow\quad (Q(u))_{1,n}> 0.$$
%i.e., the entry at position $(1,n)$ of the matrix $Q(u)$ is positive.
\end{theorem}
\begin{proof}
Suppose there exists such a collection of matrices.
Take the state set $S=\{s_1,\ldots,s_n\}$, vectors $\pi=(1,0,\ldots,0)$ 
and $f=(0,\ldots,0,1)^T$ and consider the
monoidal generalized automaton $$\cA = (S,M,\{Q(x)\mid x\in G_M\},\pi,f).$$ 
Then for each non-empty word $u =x_1\ldots x_k\in M$, by hypothesis,
$Q(u)=Q(x_1)\cdots Q(x_k)$ and furthermore
$$\pi Q(u)f= (Q(u))_{1,n}.$$
It follows that $L(\cA,0) = L\setminus\{\epsilon\}$.
By Prop.~\ref{p-0e}, $L$ is also a monoidal stochastic automaton language.
%there exists a monoidal generalized finite-state automaton $\cA'$ such that $L(\cA',0)=L$.

Conversely, let $L$ be a monoidal stochastic automaton language.
Then there exists a monoidal generalized automaton 
$$\cA_3=(S_3,M,\{Q_3(x)\mid x\in G_M\},\pi_3,f_3)$$ as given in Prop.~\ref{p-0} such that 
$L\setminus\{\epsilon\} = L(\cA_3,0)\setminus\{\epsilon\}$.

Let $S_3=\{s_1,\ldots,s_m\}$.
Define the monoidal generalized automaton
$$\cA'=(S',M,\{Q'(x)\mid x\in G_M\},\pi',f'),$$ 
where
$S' = S_3\cup \{s_0,s_{m+1}\}$, $\pi'=(1,0,\ldots,0)$, $f'=(0,\ldots,0,1)^T$, and
$$Q'(x) = \left( \begin{array}{c|c|c}
0 & \pi_3 Q_3(x) & \pi_3 Q_3(x)f_3\\\hline
0 &  & \\ 
\vdots & Q_3(x) & Q_3(x)f_3\\
0 &  & \\ \hline
0 & 0 & 0 
\end{array} \right),\quad x\in G_M.$$
Then for each non-empty word $u=x_1\ldots x_k\in M$, 
$$Q'(u) = Q'(x_1)\cdots Q'(x_k)
= \left( \begin{array}{c|c|c}
0 & \pi_3 Q_3(u) & \pi Q_3(u)f_3\\\hline
0 &  & \\ 
\vdots & Q_3(u) & Q_3(u)f_3\\
0 &  & \\ \hline
0 & 0 & 0 
\end{array} \right).$$
Thus % for each non-empty word $u=x_1\ldots x_k\in M$, %$u\ne \epsilon$,
$$\pi'Q'(u)f' = \pi_3 Q_3(u) f_3,$$ where $\pi Q'(u)f$ is the $(1,n)$-entry of the matrix $Q'(u)$.
It follows that $u\in L$ iff $(Q'(u))_{1,n}>0$ and the matrices $\{Q'(x)\mid x\in G_M\}$ have the required form.
%Hence, $L(A',0)\setminus\{\epsilon\}=L\setminus\{\epsilon\}$.
%$A_3=(S_3,M,\{Q_3(m)\mid m\in M\},\pi_3,f_3)$ and a cut point $\lambda$ such that 
%$L(A_3,0)\setminus\{\epsilon\}=L\setminus\{\epsilon\}$.
%Define the monoidal generalized acceptor
%$A'=(S',M,\{Q'(m)\mid m\in M\},\pi',f')$, where
%$S' = S_3\cup \{z',z''\}$, $\pi=(1,0,\ldots,0)$, $f'=(0,\ldots,0,1)^T$, and
%$$Q'(m) = \left(
%\begin{array}{c|c|c}
%0 & \pi_3Q_3(m) & \pi_3 Q_3(m)f_3\\\hline
%0 & Q_3(m) & Q_3(m)f_3\\\hline
%0 & 0 & 0 
%\end{array}
%\right).$$
%Note that for each word $x=m_1\ldots m_l\in U_M$, $Q'(x)$ has the same form as $Q'(m)$.
%Thus for each word $x\in U_M$ with $x\ne \epsilon$,
%$\pi'Q'(x)f' = \pi_3 Q_3(x) f_3$.
\end{proof}

%DOUBLETTE
\begin{example}\label{e-xy}
In view of Ex.~\ref{e-xy0}, consider the commutative monoid $M$ 
%given by the presentation $\langle x,y\mid xy=yx\rangle$.
%Each element of $M$ has the form $x^iy^j$ for some $i,j\geq 0$.
%Define the matrices
%$$Q(x) = \left(\begin{array}{rr} 1 & 1 \\ 0 & 1 \end{array}\right)
%\quad\mbox{and}\quad
%Q(y) = \left(\begin{array}{rr} 1 & -1 \\ 0 & 1 \end{array}\right).$$
%%where $Q(x)^{-1} = Q(y)$.
given by the words of the form $x^iy^j$, where $i,j\geq 0$.
The matrix $Q(x^iy^j)$ has the $(1,2)$-entry $i-j$.
Thus by Prop.~\ref{p-mga-L}, the corresponding monoidal generalized language $L$ 
has the non-empty words $x^iy^j$, where $i>j\geq 0$.
This language is context-free, but not regular~\cite{salomaa}.
%Moreover, if the language $L$ is specified by an automaton $\cA$ with initial vector $\pi$ and final vector $f$, 
%the value $\pi f$ will determine whether the empty word lies in~$L$.
\EXX
\end{example}

%------------------------------------------------------------
\section{Homomorphisms and Closure Properties}

Homomorphisms between monoidal generalized automata will be introduced and closure properties of 
monoidal generalized automata will be studied.
First, note that each monoidal generalized automaton $\cA=(S,M,\{Q(x)\mid x\in G_M\},\pi,f)$
can be associated with the finitely generated multiplicative matrix monoid 
$$H(\cA) =\langle Q(x)\mid x\in G_M\rangle,$$ 
which is a submonoid of $(\RR^{n\times n},\cdot,I_n)$,
where $n$ is the number of states of $\cA$.
By the extension postulate, 
\begin{eqnarray}
H(\cA) = \{Q(u)\mid u\in M\}.
\end{eqnarray}

\begin{example}
In view of Ex.~\ref{e-xy0}, the multiplicative matrix monoid $H$ is generated by the matrices
$$\left(\begin{array}{rr} 1 & 1 \\ 0 & 1 \end{array}\right) \quad\mbox{and}\quad
\left(\begin{array}{rr} 1 & -1 \\ 0 & 1 \end{array}\right).$$
Thus the matrix monoid (group) is
$$H = \left\{ \left(\begin{array}{rr} 1 & k \\ 0 & 1 \end{array}\right) \mid k\in \ZZ \right\},$$
which is isomorphic to $(\ZZ,+,0)$ by the isomorphism $Q\mapsto (Q)_{1,2}$.
\EXX
\end{example}

Let $\cA=(S,M,\{Q(x)\mid x\in G_M\},\pi,f)$ and $\cA'=(S,M',\{Q'(x')\mid x'\in G_{M'}\},\pi,f)$ 
be monoidal generalized automata.
Then $\cA'$ is a {\em homomorphic image} of $\cA$ if
there exist monoid homomorphisms
\begin{eqnarray}\label{e-uq0}
\phi:M\rightarrow M':x\mapsto x' \quad \mbox{and}\quad \psi:H(\cA)\rightarrow H(\cA'):Q\mapsto Q' 
\end{eqnarray}
such that the following {\em commuting property\/} holds for all $u\in M$,
\begin{eqnarray}\label{e-uq}
Q'(u') = Q(u)',
\end{eqnarray}
i.e., the following diagram commutes
\begin{center}
\mbox{$
\xymatrix{
M\ar@{->}[rr]^\phi\ar@{->}[d]_Q && M'\ar@{->}[d]^{Q'} \\
H\ar@{->}[rr]^\psi && H'\\
}
$}
\end{center}

\begin{proposition}\label{p-comm}
Let $\phi:M\rightarrow M'$ be an epimorphism.
Then the extension postulate for the homomorphic image automaton $\cA'$ is a consequence 
of the extension postulate for the automaton $\cA$ and the commuting property. 
\end{proposition}
\begin{proof}
By hypothesis, the collection of matrices $\{Q'(x')\mid x'\in G_{M'}\}$ of $\cA'$ is fully given by 
the collection of matrices $\{Q(x)\mid x\in G_{M}\}$ of $\cA$, where $Q'(x') = Q(x)'$ for each $x\in G_M$.
Moreover, for each word $u=x_1\ldots x_k\in M$, $Q'(u') = Q(u)'$ and therefore
\begin{eqnarray*}
Q'(u') &=& Q'((x_1\ldots x_k)') = Q'(x'_1\ldots x'_k)\\
&=& Q(x_1\ldots x_k)' = (Q(x_1)\cdots Q(x_k))'\\
&=& Q(x_1)'\cdots Q(x_k)' = Q'(x'_1)\cdots Q'(x'_k).
\end{eqnarray*}
%and so by the commuting property $Q'(u') = Q(u)'$ provides the unique extension.
\end{proof}

\begin{example}
Consider the commutative monoid $M$ given by the presentation $\langle x,y\mid xy=yx\rangle$ 
and define the matrices
$$Q(x) = \left(\begin{array}{rr} 1 & 1 \\ 0 & 1 \end{array}\right)
\quad\mbox{and}\quad
Q(y) = \left(\begin{array}{rr} 1 & -1 \\ 0 & 1 \end{array}\right).$$
By Prop.~\ref{p-mga-L}, 
the corresponding monoidal generalized automaton language has the non-empty words $x^iy^j$, where $i>j\geq 0$.
%\begin{itemize}
%\item

Define the monoid $M'$ as the homomorphic image of the monoid homomorphism $\phi:M\rightarrow M'$, 
where
$$x'=\phi(x)=x^2\quad\mbox{and}\quad y'=\phi(y)=y^2,$$
and define the associated matrices as
$$Q'(x') = \left(\begin{array}{rr} 1 & 2 \\ 0 & 1 \end{array}\right)
\quad\mbox{and}\quad
Q'(y') = \left(\begin{array}{rr} 1 & -2 \\ 0 & 1 \end{array}\right).$$
The matrix homomorphism is given by squaring, 
$$Q(x)\mapsto Q(x)' = Q(x)^2 \quad\mbox{and}\quad Q(y)\mapsto Q(y)' = Q(y)^2.$$
It is easy to check that the commuting property holds: $Q'(u') = Q(u)'$ for all $u\in M$.
By Prop.~\ref{p-mga-L}, 
the monoidal generalized automaton language
has the non-empty words ${x'}^i{y'}^j = x^{2i}y^{2j}$, where $i>j\geq 0$.
%\item
%Second, define the monoid $M'$ as the homomorphic image of the homomorphism $\phi:M\rightarrow M'$, 
%where 
%$$x'=\phi(x)=x \quad \mbox{and} \quad y'=\phi(y)=\epsilon,$$ 
%and define the associated matrices as
%$$Q'(x') = \left(\begin{array}{rr} 1 & 1 \\ 0 & 1 \end{array}\right)
%\quad\mbox{and}\quad
%Q'(y') = \left(\begin{array}{rr} 1 & 0 \\ 0 & 1 \end{array}\right).$$
%The matrix homomorphism is then given by the assignment
%$$Q(x)\mapsto Q(x) = Q'(x')\quad\mbox{and}\quad Q(y)\mapsto I_2 = Q'(y').$$
%It is easy to check that the commuting property holds: $Q'(u') = (Q(u))'$ for all $u\in M$.
%By Prop.~\ref{p-mga-L}, 
%the monoidal generalized automaton language
%has the non-empty words ${x'}^i = x^i$, where $i>0$.
%\end{itemize}
\EXX
\end{example}

%\begin{example}
%Consider the commutative monoid $M$ given by the presentation $\langle x,y\mid xy=yx\rangle$ 
%and define the matrices
%$$Q(x) = \left(\begin{array}{rr} 1 & 1 \\ 0 & 1 \end{array}\right)
%\quad\mbox{and}\quad
%Q(y) = \left(\begin{array}{rr} 1 & -1 \\ 0 & 1 \end{array}\right).$$
%By Prop.~\ref{p-mga-L}, 
%the corresponding monoidal generalized automaton language has the non-empty words 
%$x^iy^j$, where $i>j\geq 0$.
%
%Define the monoid $M'$ as the homomorphic image of the homomorphism $\phi:M\rightarrow M'$, 
%where 
%$$x'=\phi(x)=x \quad \mbox{and} \quad y'=\phi(y)=\epsilon,$$ 
%and define the associated matrices as
%$$Q'(x') = \left(\begin{array}{rr} 1 & 1 \\ 0 & 1 \end{array}\right)
%\quad\mbox{and}\quad
%Q'(y') = \left(\begin{array}{rr} 1 & 0 \\ 0 & 1 \end{array}\right).$$
%The matrix homomorphism is then given by the assignment
%$$Q(x)\mapsto Q(x) = Q'(x')\quad\mbox{and}\quad Q(y)\mapsto I_2 = Q'(y').$$
%However, the commuting property does not hold, since 
%$$Q(xy)' = (Q(x)Q(y))' = I_2' = I_2$$ 
%and
%$$Q'((xy)') = Q'(x'y') = Q'(x') = Q(x).$$
%\EXX
%\end{example}

\begin{proposition}[Classical Generalized Languages]
Each monoidal generalized automaton is a %proper 
homomorphic image of a classical generalized automaton.  
Each monoidal generalized automaton language can be obtained as a homomorphic image of a classical generalized automaton 
language.
\end{proposition}
\begin{proof}
Let $\cA'=(S,M,\{Q'(x')\mid x'\in G_M\},\pi,f)$ be a monoidal generalized automaton over the monoid $(M,\circ,e)$.
Define the alphabet
$$\Sigma = \{ x\mid x'\in G_M\}$$
and consider the mapping $\phi_0:\Sigma\rightarrow M$ defined by $\phi_0(x)=x'$.
Since $\Sigma^*$ is a free monoid, 
there exists a unique extension of $\phi_0$ to a monoid homomorphism $\phi:\Sigma^*\rightarrow M$.

Let $\cA=(S,\Sigma,\{Q(x)\mid x\in \Sigma\},\pi,f)$ be the classical generalized automaton,
where the matrices are defined as $Q(x) = Q'(x')$ for each $x\in \Sigma$. 
%i.e., the monoid homomorphism $\psi$ between the matrix monoids is the identity mapping.
The extension postulate holds in $\cA$, since the underlying monoid is free. %monoid $\Sigma^*$.

In view of the extension postulates, for each word $u=x_1\ldots x_k\in \Sigma^*$,
\begin{eqnarray*}
Q'(u') 
&=& Q'(x'_1\ldots x'_k) = Q'(x'_1) \cdots Q'(x'_k) \\
&=& Q(x_1)\cdots Q(x_k) = Q(x_1\ldots x_k)= Q(u).
\end{eqnarray*}
By the setting $Q(u)' = Q(u)$, the commuting property holds.
%Moreover, the commuting property holds, since for each word $u=x_1\ldots x_k\in \Sigma^*$,
%\begin{eqnarray*}
%(Q(u))' 
%&=& (Q(x_1\ldots x_k))' = (Q(x_1)\cdots Q(x_k))' \\
%&=& (Q(x_1))'\cdots (Q(x_k))' = Q'(x'_1) \cdots Q'(x'_k) \\ 
%&=& Q'(x'_1\ldots x'_k) =  Q'(u').
%\end{eqnarray*}
Therefore $\pi Q(u)f = \pi Q'(u') f$ and 
so $L(\cA',\lambda) =\phi(L(\cA,\lambda))$ for each cut point $\lambda$.
%i.e., $\cA'$ is a proper homomorphic image of $\cA$.
\end{proof}
The reader may check Ex.~\ref{e-4}.
%The classical generalized finite-state automaton $A'$ defined above is the {\em free companion} of $A$.

\begin{proposition}[Set Operations]
The class of monoidal generalized languages is closed 
under union, intersection and complement with regular monoidal languages.
\end{proposition}
\begin{proof}
Let $L$ be a monoidal generalized language and let $R$ be a regular monoidal language over a common monoid $M$.
Then there exists a monoidal stochastic automaton
$$\cA=(S,M,\{Q(x)\mid x\in G_M\},\pi,f)$$ 
and a monoidal automaton 
$$\cA' = (S',M,I,F,\Delta)$$
such that $L=L(\cA,\lambda)$ for some cut point $0\leq \lambda\leq 1$ and $R=L(\cA')$.
By Prop.~\ref{p-malg},
it may be assumed that $\cA'$ is a monoidal stochastic automaton
$$\cA'=(S',M,\{Q'(x)\mid x\in G_M\},\pi',f')$$ 
such that $R=L(\cA',0)$, where
the matrices $Q'(x)$ have one entry~1 in each row and the initial state distribution $\pi'$ 
has a single entry~1.

Let $S=\{s_1,\ldots,s_n\}$ and $S'=\{s_{n+1},\ldots,s_{n+m}\}$.
Define the monoidal stochastic automaton
$$\cA''=(S'',M,\{Q''(x)\mid x\in G_M\},\pi'',f''),$$ 
where $S'' = S\cup S'$, $\pi'' = \frac{1}{2}(\pi,\pi')$,
$f ={f\choose f'}$ and
$$Q''(x) =  \left(\begin{array}{cc} Q(x) & 0 \\ 0 & Q'(x)\end{array}\right),\quad x\in G_M.$$
Then for all words $u=x_1\ldots x_k\in M$,
$$Q''(u) = \left(\begin{array}{cc} Q(u) & 0 \\ 0 & Q'(u)\end{array}\right).$$
Thus %for each word $u=x_1\ldots x_k\in M$,
$$\pi''Q''(u) f'' = \frac{1}{2} \pi Q(u) f + \frac{1}{2} \pi' Q'(u) f'.$$ 
By the structure of $Q'(u)$ and $\pi'$, the term $\pi' Q'(u) f'$ is either~0 or~1.
Therefore, $u\in L\cup R$ iff 
$\frac{1}{2} \pi Q(u) f > \frac{1}{2} \lambda$ 
or 
$\frac{1}{2} \pi' Q'(u) f' = \frac{1}{2}.$ 

If $\lambda=1$, then $L=\emptyset$ and so $L\cup R=R$.
Otherwise, $\frac{1}{2}>\frac{1}{2} \lambda$.
%if $u\in L\cup R$, then 
%$\frac{1}{2} \pi Q(u) f > \frac{1}{2} \lambda$ 
%and
%$\frac{1}{2} \pi' Q'(u) f' = \frac{1}{2}$, 
%or 
%$\frac{1}{2} \pi Q(u) f \leq \frac{1}{2} \lambda$ 
%and
%$\frac{1}{2} \pi' Q'(u) f' = \frac{1}{2} >\frac{1}{2}\lambda$, 
%or
%$\frac{1}{2} \pi Q(u) f > \frac{1}{2} \lambda$ 
%and
%$\frac{1}{2} \pi' Q'(u) f' = \frac{1}{2} >\frac{1}{2}\lambda$. 
%In each case, $\pi''Q''(u)f'' > \frac{1}{2} \lambda$.
%Conversely,
It easily follows that $u\in L\cup R$ iff $\pi''Q''(u)f''> \frac{1}{2} \lambda$ and hence
$L\cup R = L(\cA'', \frac{1}{2} \lambda)$.

Similarly, $L\cap R = L(\cA'', \frac{1}{2} (\lambda+1))$ and thus
$L\cap R$ is also a monoidal stochastic automaton language.
Moreover, $L\setminus R = L\cap \overline R$, where $\overline R$ is the set complement of $R$.
Since $\overline R$ is also regular, $L\setminus R$ is also a monoidal stochastic automaton language.
\end{proof}

%CHECK
\begin{proposition}[Closure under Complement]
The class of monoidal automaton languages corresponding to monoidal generalized automata 
with isolated cut points is closed under complement.
\end{proposition}

%The definition of the monoidal generalized finite-state automaton can be modified to allow
%anti-homomorphisms in the extension postulate~(\ref{e-post0}), i.e.,
%for each word $u=x_1\ldots x_k\in M$,
%\begin{eqnarray}\label{e-post00}
%Q(u) = Q(x_k) \cdots Q(x_1) \quad\mbox{and}\quad Q(\epsilon) = I_n.
%\end{eqnarray}
%Moreover, the definition of the homomorphic image can be altered to allow
%anti-homomorphisms~(\ref{e-uq0}).
%The commuting property will then change accordingly.
%%This will be used in the following result.
We may allow anti-homomorphisms in the definition of homomorphic images. %~(\ref{e-uq0}).
An {\em anti-homomorphism} is a mapping $\phi:M\rightarrow M'$ between two monoids such that
$\phi(xy) = \phi(y)\phi(x)$ for all $x,y\in M$.
The commuting property may then change accordingly.

\begin{proposition}[Mirror Images]
The monoidal generalized automaton given by the mirror image of a 
monoidal generalized automaton $\cA$ is a %proper 
homomorphic image of\/ $\cA$.
The class of monoidal generalized automata languages is closed under mirror images.
\end{proposition}
\begin{proof}
Let $L$ be a monoidal generalized language.
Then there exists a monoidal generalized automaton 
$$\cA=(S,M,\{Q(x)\mid x\in G_M\},\pi,f)$$ and a real number $\lambda$
such that $L=L(\cA,\lambda)$.

Define the monoidal generalized automaton $$\cA'=(S,M,\{Q(x)^T\mid x\in G_M\},f^T,\pi^T),$$
which is the homomorphic image of $\cA$ under the monoid anti-homomorphism 
given by taking the mirror image
$$\phi:M\rightarrow M:x=x_1\ldots x_k\mapsto x'=x_k\ldots x_1$$
and the matrix monoid anti-homomorphism given by transposition
$$\psi:H(\cA)\rightarrow H(\cA'):Q\mapsto Q'=Q^T.$$ 

Then for each word $u=x_1\ldots x_k\in M$, 
\begin{eqnarray*}
Q(u)^T &=& Q(x_1\ldots x_k)^T = (Q(x_1)\cdots Q(x_k))^T \\
&=& Q(x_k)^T\cdots Q(x_1)^T %= Q(x_k)'\cdots Q(x_1)'
\end{eqnarray*}
Thus 
\begin{eqnarray*}
\pi Q(u) f 
&=& (\pi Q(u) f)^T = f^T Q(u)^T \pi^T = f^T (Q(x_1)\cdots Q(x_k))^T\pi^T \\
&=& f^T  Q(x_k)^T\cdots Q(x_1)^T \pi^T = f^T Q(u)^T \pi^T.
\end{eqnarray*}
Hence, $L(\cA',\lambda) = \phi(L(\cA,\lambda))$ for each cut point $\lambda$.
%i.e., $\cA'$ is a proper homomorphic image of $\cA$.
\end{proof}

%\begin{proposition}
%The class of monoidal generalized languages with isolated cut points is closed under set complement.
%\end{proposition}
%\begin{proof}
%Let $L=L(A,\lambda)$ be a monoidal generalized language accepted by a monoid generalized automaton
%$A=(S,M,\{Q(m)\mid m\in U_M\},\pi,f)$ w.r.t.\ the isolated cut point $\lambda$.
%Then the complement of $L$ is given by
%\begin{eqnarray*}
%\overline L &=& \{x\mid \forall m_1,\ldots,m_k\in U_M: x=m_1\ldots m_k\Rightarrow \pi Q(x) f \leq \lambda\}.
%\end{eqnarray*}
%This set splits into two subsets, $\{x\in \overline L \mid \pi Q(x) f = \lambda\}$ and 
%$\{x\in \overline L \mid \pi Q(x) f < \lambda\}$.
%Since the cut point $\lambda$ is isolated, the set
%$\{x\in \overline L \mid \pi Q(x) f = \lambda\}$ is empty.
%Moreover, we have $\pi Q(x) f < \lambda$ iff $(-\pi) Q(x) f > -\lambda$.
%Then the monoidal generalized automaton $A'=(S,M,\{Q(m)\mid m\in U_M\},-\pi,f)$ 
%w.r.t.\ the cut point $-\lambda$ accepts the language
%$$\{x\in U_M \mid \exists m_1,\ldots,m_k\in U_M: x=m_1\ldots m_k\wedge (-\pi) Q(x) f > - \lambda\}.$$
%Is this language equal to $\overline L$.
%%???????????????????????????????????
%
%Note that if $\lambda$ is an isolated cut point for $A$, 
%i.e., for some $\delta>0$, $|\pi Q(x) f-\lambda|>\delta$ for all $x\in U_M$,
%then $|(-\pi) Q(x) f-(-\lambda)| = |\pi Q(x) f-\lambda|>\delta$ for all $x\in U_M$ 
%and so $-\lambda$ is an isolated cut point for $A'$.
%\end{proof}

%Let $M_i = (M_i,\circ_i,e_i)$ be monoids, $1\leq i\leq n$.
A {\em monoidal generalized $2$-tape automaton} is a monoidal generalized automaton 
$\cA=(S,M,\{Q(x)\mid x\in G_{M}\},\pi,f)$
over the monoid $M = M_1\times M_2$, which is the Cartesian product of two monoids.
A {\em monoidal generalized $2$-tape automaton language} is a monoidal generalized automaton language recognized 
by a monoidal generalized $2$-tape automaton.
This notion can be extended to monoidal generalized $n$-tape automata and monoidal generalized $n$-tape
automata languages for $n\geq 2$.

The class of monoids is closed under Cartesian products and 
thus the monoidal generalized $n$-tape automata are
a special case of the monoidal generalized automata.

Let $M_1 = (M_1,\circ_1,e_1)$ and $M_2 = (M_2,\circ_2,e_2)$ be monoids.
In view of the monoidal generalized $2$-tape automaton
$$\cA=(S,M,\{Q(x)\mid x\in G_{M}\},\pi,f)$$ over the Cartesian product monoid $M=M_1\times M_2$,
it is assumed that the generating set $G_M$ is defined 
component-wise in terms of generating sets $G_{M_1}$ and $G_{M_2}$ for $M_1$ and $M_2$, respectively; i.e., 
$$G_M = \{(x_1,e_2)\mid x_1\in G_{M_1}\} \cup \{(e_1,x_2)\mid x_2\in G_{M_2}\}.$$
That is, the matrices $Q(x)$, $x\in G_M$, have the form
$Q(x_1,e_2)$ for all $x_1\in G_{M_1}$ and $Q(e_1,x_2)$ for all $x_2\in G_{M_2}$. 
Then the extension postulate for $\cA$ states that for all words 
$u=x_1\ldots x_k\in M_1$ and $v=y_1\ldots y_l\in M_2$, 
\begin{eqnarray}\label{e-u,v}
Q(u,v) &=& Q(u,e_2)Q(e_1,v) \\ &=& (Q(x_1,e_2)\cdots Q(x_k,e_2))(Q(e_1,y_1)\cdots Q(e_1,y_l)). \nonumber
\end{eqnarray}
This definition can be extended in a straightforward manner to monoidal generalized $n$-tape automata.

\begin{example}\label{e-2tape-m}
Consider the commutative monoid $M_1$, given by the presentation $\langle x,y\mid xy=yx\rangle$, and 
the monoid $M_2=(\NN_0,+,0)$.
Define the associated matrices
$$
Q(x,0) = Q(y,0) =\left(\begin{array}{rr} 1 & 1 \\ 0 & 1 \end{array}\right)
\quad\mbox{and}\quad
Q(e,1) = \left(\begin{array}{rr} 1 & -1 \\ 0 & 1 \end{array}\right).
$$
Each element of the Cartesian product monoid $M = M_1\times M_2$ has the form $(x^iy^j,k)$, where $i,j,k\geq 0$.
The associated matrix $Q(x^iy^j,k)$ has the $(1,2)$-entry $i+j-k$.
By Prop.~\ref{p-mga-L}, 
the corresponding monoidal generalized 2-tape automaton language $L$ has the non-empty words 
$(x^iy^j,k)$, where $i+j>k$.
\EXX
\end{example}

\begin{proposition}[Inverse Relations]
The monoidal generalized automaton given by the inverse relation of a monoidal generalized automaton $\cA$ is 
a homomorphic image of\/ $\cA$.
The class of monoidal generalized 2-tape languages is closed under inverse relations.
\end{proposition}
\begin{proof}
Let $\cA = (S,M_1\times M_2,\{Q(x,y)\mid (x,y)\in G_{M_1\times M_2}\},\pi,f)$ 
be a monoi\-dal generalized 2-tape automaton.

Define the monoidal generalized 2-tape automaton 
$$\cA' = (S,M_2\times M_1,\{Q'(y,x)\mid (y,x)\in G_{M_2\times M_1}\},\pi,f),$$ 
which is the homomorphic image of $\cA$ under the monoid homomorphism
$$\phi:M_1\times M_2\rightarrow M_2\times M_1:(u,v) \mapsto (v,u)$$ 
and the monoid homomorphism 
$$\phi:H(\cA)\rightarrow H(\cA'): Q(u,v) \mapsto Q'(v,u) =Q (u,v).$$ 

%The automaton $\cA'$ satisfies the extension postulate,
%since % under the involution given by transposing components.
%for all words $(u,v)\in M_1\times M_2$ with $u=x_1\ldots x_k$ and $v=y_1\ldots y_l$, 
%\begin{eqnarray*}
%Q'(v,u) 
%&=&  Q(u,v) \\
%&=& Q(x_1,e_1)\ldots Q(x_k,e_1) Q(e_2,y_1)\ldots Q(e_2,y_l)\\
%&=& Q'(e_1,x_1)\ldots Q'(e_1,x_k) Q'(y_1,e_2)\ldots Q'(y_l,e_2).
%\end{eqnarray*}
%Moreover, the commuting property holds, since for all words $(u,v)\in M_1\times M_2$ with
%$$Q'((u,v)') = Q'(v,u) = Q(u,v) = (Q(u,v))'.$$
The automaton $\cA'$ satisfies the commuting property, since for all words $(u,v)\in M_1\times M_2$,
$$Q'((u,v)') = Q'(v,u) = Q(u,v) = Q(u,v)'.$$
Therefore, by Prop.~\ref{p-comm}, the automaton $\cA'$ fulfills the extension postulate.

Thus for each pair $(u,v)\in M_1\times M_2$, 
$$\pi Q'(v,u) f = \pi Q(u,v)' f= \pi Q(u,v) f.$$
Hence, $L(\cA',\lambda) = \phi(L(\cA,\lambda))$ for each cut point $\lambda$.
%i.e., $\cA'$ is a proper homomorphic image of $\cA$.
\end{proof}

\begin{example}
Reconsider the Cartesian product monoid $M=M_1\times M_2$ from Ex.~\ref{e-2tape-m}.
The inverse relation leads to the monoid $M'=M_2\times M_1$ with the associated matrices
$$
Q(0,x) = Q(0,y) =\left(\begin{array}{rr} 1 & 1 \\ 0 & 1 \end{array}\right)
\quad\mbox{and}\quad
Q(1,e) = \left(\begin{array}{rr} 1 & -1 \\ 0 & 1 \end{array}\right).
$$
%$$Q(1,a) = \left(\begin{array}{rr} 1 & 1 \\ 0 & 1 \end{array}\right)
%\quad\mbox{and}\quad
%Q(1,b) = \left(\begin{array}{rr} 1 & -1 \\ 0 & 1 \end{array}\right).$$
By Prop.~\ref{p-mga-L}, 
the corresponding monoidal generalized 2-tape automaton language has the non-empty words $(k,x^iy^j)$, where $i+j>k$.
\EXX
\end{example}

\begin{proposition}[Cartesian Products]
The class of monoidal generalized automa\-ton languages is closed under Cartesian products.
\end{proposition}
\begin{proof}
Let $L_1$ and $L_2$ be two monoidal generalized languages.
Then there exist monoidal generalized automata
$\cA_1 = (S_1,M_1,\{Q_1(x)\mid x\in G_{M_1}\}, \pi_1,f_1)$ and
$\cA_2 = (S_2,M_2,\{Q_2(x)\mid x\in G_{M_2}\}, \pi_2,f_2)$ such that
$L_1=L(A_1,\lambda_1)$ and $L_2=L(A_2,\lambda_2)$ for some cut points $\lambda_1$ and $\lambda_2$.
In view of the results in Sect.~\ref{s-5}, 
the matrices, initial vectors, final vectors and cut points can be chosen to be non-negative.
%By Prop.~\ref{p-0}, we may assume that $\lambda_1=0=\lambda_2$.

The following construction makes use of the Kronecker product of matrices.
For this, consider the monoidal generalized automaton 
$$\cA = (S,M,\{Q(y)\mid y\in G_{M}\}, \pi,f),$$
where $S=S_1\times S_2$,
$M = M_1\times M_2$,
$\pi = \pi_1\otimes \pi_2$, $f = f_1\otimes f_2$ and
\begin{eqnarray*}
Q(x_1,e_2) &=& Q_1(x_1)\otimes Q_2(e_2), \\
Q(e_1,x_2) &=& Q_1(e_1)\otimes Q_2(x_2)
\end{eqnarray*}
for all $x_1\in G_{M_1}$ and $x_2\in G_{M_2}$,
and $Q_1(e_1)=I$ and $Q_2(e_2)=I'$ are identity matrices.
Since the Kronecker product is a bilinear form, we obtain
for all $u\in M_1$ and $v\in M_2$,
%for all $u=x_1\ldots x_k\in M_1$ and $v=y_1\ldots y_l\in M_2$,
%\begin{eqnarray*}
%Q(u,e_2) &=&  (Q_1(x_1)\cdots Q_1(x_k)) \otimes Q_2(e_2),\\
%Q(e_1,v) &=&  Q_1(e_1)\otimes (Q_2(y_1)\cdots Q_2(y_l)).
%\end{eqnarray*}
%Then 
%$$Q(u,v) = Q(u,e_2)\cdot Q(e_1,v)$$
%and thus the automaton $\cA$ satisfies the extension postulate.
%
%Moreover, for all words $u\in M_1$ and $v\in M_2$,
%\begin{eqnarray*}
%Q(u,v) 
%&=& Q(u,e_2)\cdot Q(e_1,v) \\
%%&=& (Q_1(x_1,e_2)\cdots Q_1(x_k,e_2)) (Q_2(e_1,y_1)\cdots Q_2(e_1,y_l))\\
%&=& ((Q_1(x_1)\otimes I')\cdots (Q_1(x_k)\otimes I')) ((I\otimes Q_2(y_1))\cdots (I\otimes Q_2(y_l)))\\
%%&=& ((Q_1(x_1)\cdots Q_1(x_k))\otimes I') (I\otimes (Q_2(y_1)\cdots Q_2(y_l)))\\
%&=& (Q_1(u)\otimes I')\cdot (I\otimes Q_2(v))\\
%&=& Q_1(u)\otimes Q_2(v).
%\end{eqnarray*}
\begin{eqnarray*}
Q(u,v) &=& Q_1(u)\otimes Q_2(v).
\end{eqnarray*}
Thus
\begin{eqnarray*}
\pi Q(u,v) f 
&=& (\pi_1\otimes \pi_2) (Q_1(u)\otimes Q_2(v)) (f_1\otimes f_2)\\
&=& (\pi_1Q_1(u)f_1)\otimes (\pi_2Q_2(v)f_2)\\
&=& (\pi_1Q_1(u)f_1)\cdot (\pi_2Q_2(v)f_2),
\end{eqnarray*}
since the components of the last tensor product are scalars.
Therefore, by the hypothesis on non-negativity,
$\pi_1Q_1(u)f_1> \lambda_1$ and $\pi_2Q_2(v)f_2>\lambda_2$ iff $\pi Q(u,v) f > \lambda_1\lambda_2$.
%The result follows.
\end{proof}

%\begin{example}
%Reconsider the monoids $M_1$ and $M_2$ from Ex.~\ref{e-2tape-m}.
%The construction of Cartesian product leads to the matrices
%\EXX
%\end{example}

\begin{example}
Let $m\geq 2$. 
Consider the $m$-adic acceptor 
$$\cA_1 = (\{s_1,s_2\},\{0,\ldots,m-1\},P_1,\pi_1,f_1)$$ 
and the stochastic automaton 
$$\cA_2 = (\{s'_1,s'_2\},\{y\},P_2,\pi_2,f_2),$$ 
where
$$P_2(y) = \left(\begin{array}{cc} \frac{1}{2} & \frac{1}{2}\\ 0 & 1 \end{array}\right),$$
$\pi_2=(1,0)$ and $f_2 = (0,1)^T$.
Since for each integer $k\geq 0$,
$$P_2(y^k) = \left(\begin{array}{cc} \frac{1}{2^k} & 1-\frac{1}{2^k}\\ 0 & 1 \end{array}\right),$$
the accepted language is 
$L(\cA_2,\lambda) = \{y^i\mid i\geq k\}$, 
where $1-\frac{1}{2^{k-1}}\leq \lambda< 1-\frac{1}{2^k}$.

The construction of the Kronecker product leads to the monoidal generalized automaton
$$\cA = (\{s_1,s_2\}\times\{s'_1,s'_2\},\{0,\ldots,m-1\}\times \{y\},P,\pi,f),$$ 
where for each $0\leq x\leq m-1$,
\begin{eqnarray*}
P(x,\epsilon)
&=& P_1(x)\otimes I_2 
= \left(\begin{array}{cccc} 
1-\frac{x}{m}   & \frac{x}{m}   & 0 & 0 \\
1-\frac{x+1}{m} & \frac{x+1}{m} & 0 & 0 \\
0 & 0 & 1-\frac{x}{m}   & \frac{x}{m}   \\
0 & 0 & 1-\frac{x+1}{m} & \frac{x+1}{m} \\
\end{array}\right),\\
P(\epsilon, y) 
&=& I_2\otimes P_2(y)
= \left(\begin{array}{cccc} 
\frac{1}{2}   & \frac{1}{2}   & 0 & 0 \\
0             & 1             & 0 & 0 \\
0 & 0 & \frac{1}{2}   & \frac{1}{2}  \\
0 & 0 & 0 & 1
\end{array}\right),
\end{eqnarray*}
$\pi= (1,0,0,0)$ and $f=(0,0,0,1)^T$.
Then
\begin{eqnarray*}
P(x,y) 
&=& P_1(x)\otimes P_2(y)\\
&=& \left(\begin{array}{cccc} 
\frac{1}{2}(1-\frac{x}{m}) & \frac{1}{2}\frac{x}{m} & \frac{1}{2}(1-\frac{x}{m}) & \frac{1}{2}\frac{x}{m} \\
\frac{1}{2}(1-\frac{x+1}{m}) & \frac{1}{2}\frac{x+1}{m} & \frac{1}{2}(1-\frac{x+1}{m}) & \frac{1}{2}\frac{x+1}{m} \\
0 & 0 & 1-\frac{x}{m} & \frac{x}{m} \\
0 & 0 & 1-\frac{x+1}{m} & \frac{x+1}{m} \\
\end{array}\right).
\end{eqnarray*}

Let $u=x_1\ldots x_k\in\Sigma^*$ and $i\geq 0$ be an integer.
The word $(u,y^i)$ lies in $L(\cA,\lambda)$ iff $\pi P(u,y^i) f>\lambda$, i.e.,
the $(1,4)$-entry of the matrix $P(u,y^i)$ is larger than $\lambda$.
By the Kronecker product, 
this entry is given by the product of the $(1,2)$-entries of $P_1(u)$ and $P(y^i)$,
which is $0.x_k\ldots x_1\cdot \left(1-\frac{1}{2^i}\right)$.
\EXX
\end{example}

\end{document}